\documentclass[11pt]{article}

\usepackage[margin=1in]{geometry}
\usepackage{amsfonts}
\usepackage{latexsym,amssymb,amsfonts,graphicx}
\usepackage{amsmath,dsfont}
\usepackage[linewidth=1pt]{mdframed}
\usepackage{mathrsfs}
\usepackage[normalem]{ulem}
\usepackage{epsfig}
\usepackage{tabularx}
\usepackage{verbatim}
\usepackage{enumitem}
\usepackage{framed}
\usepackage{algorithm2e}
\usepackage{float}
\usepackage{url}
\usepackage{bm}
\usepackage{color}
\usepackage{natbib}

\usepackage[unicode=true,colorlinks,citecolor=linkcolor,linkcolor=blue,urlcolor=blue]{hyperref}

\providecommand{\U}[1]{\protect\rule{.1in}{.1in}}

\newtheorem{thm}{Theorem}[section]

\newtheorem{Assumption}{\bf Assumption}

\newtheorem{corollary}{Corollary}[section]

\newtheorem{lem}{Lemma}[section]

\newtheorem{prop}{Proposition}[section]
\newtheorem{rem}{Remark}[section]

\usepackage[bf,SL,BF]{subfigure}

\newenvironment{proof}[1][Proof]{\noindent\textbf{#1.} }{\ \rule{0.5em}{0.5em}}
\numberwithin{equation}{section}
\allowdisplaybreaks[4]

\definecolor{linkcolor}{rgb}{0,0,0.7}

\begin{document}

\title{Optimal Portfolio with Ratio Type Periodic Evaluation under Short-Selling Prohibition}
\author{Wenyuan Wang\thanks{School of Mathematics and Statistics, Fujian Normal University, Fuzhou, 350007, China; and School of Mathematical Sciences, Xiamen University, Xiamen, 361005, China. Email: wwywang@xmu.edu.cn}
\and
Kaixin Yan\thanks{School of Mathematical Sciences, Xiamen University, Xiamen, 361005, China. Email: kaixinyan@stu.xmu.edu.cn}
\and
Xiang Yu\thanks{Department of Applied Mathematics, The Hong Kong Polytechnic University, Kowloon, Hong Kong. E-mail: xiang.yu@polyu.edu.hk}
}

\date{\ }

\maketitle
\vspace{-0.4in}

\begin{abstract}
This paper studies some unconventional utility maximization problems when the ratio type relative portfolio performance is periodically evaluated over an infinite horizon. Meanwhile, the agent is prohibited from short-selling stocks. Our goal is to understand the impact of the periodic reward structure on the long-run constrained portfolio strategy. For power and logarithmic utilities, we can reformulate the original problem into an auxiliary one-period optimization problem. To cope with the auxiliary problem with no short-selling, the dual control problem is introduced and studied, which gives the characterization of the candidate optimal portfolio within one period. With the help of the results from the auxiliary problem, the value function and the optimal constrained portfolio for the original problem with periodic evaluation can be derived and verified, allowing us to discuss some financial implications under the new performance paradigm.    
 
 
\end{abstract}

\vspace{0.2in}
\section{Introduction}

The classical portfolio management via utility maximization has two essential ingredients: the predetermined time horizon $T$ and the prescribed utility function capturing the agent's risk aversion. Since the pioneer studies in \cite{Mer69} and \cite{Mer71}, the finite-horizon utility maximization problems on terminal wealth have been extensively studied in the vast literature by featuring market incompleteness, trading constraints, the mixture with other decision making such as intermediate consumption or optimal stopping, among others. 

On the other hand, for the portfolio management by institutional managers such as mutual funds, the performance is often evaluated in a long run fashion instead of a short-dated horizon, giving rise to various long-term horizon problem formulations on investment decision making. For example, the Kelly's criterion, also called the optimal long-run growth rate with expected logarithmic utility, has been a well-known performance measure for the infinite horizon portfolio management. To better account for diverse risk attitudes by different agents, \cite{BPK99} and \cite{FlemingS99} independently introduced and studied the optimal long-run growth rate of expected utility as risk-sensitive stochastic control problems. Later, \cite{Pham03} incorporated the stochastic benchmark and studied the long-run outperformance criterion as a large deviation probability control problem. From a different perspective, \cite{MZ08}, \cite{MZ09}, \cite{MZ10} and many subsequent studies investigated the possibility of rolling forward the utility function in a time-consistent manner. In their studies, to overcome the limits in the static choice of the risk preference and the time horizon, it has been proposed to replace the conventional utility function by the forward performance measure, which can give guidance in the portfolio management in a possibly infinite horizon situation. 

Recently, \cite{TZ23} proposed another new criterion for the infinite horizon portfolio management featuring the periodic evaluation, where their utility is S-shaped and is generated by the difference between the current wealth level and the self-generated benchmark defined as the wealth level at the previous evaluation date. In \cite{TZ23}, the periodic evaluation is conducted over an infinite horizon and the goal of the agent is to maximize the expected total sum of all discounted utilities. This unconventional performance criterion not only can provide the optimal portfolio to meet the long run purpose, but also can capture the needs of the benchmark-based periodic performance evaluation that have been exercised in practice, such as in the annual appraisal review in the fund industry. However, due to the structure of the S-shaped utility in \cite{TZ23}, the issue of bankruptcy may become inevitable in some circumstances and the portfolio decision has to be ceased once the bankruptcy occurs at some future time. In the present paper, we are interested in extending the periodic evaluation in \cite{TZ23} by considering its variation when the relative performance is generated by the ratio of wealth processes at two evaluation dates. To simplify the mathematical analysis, we only focus on power and logarithmic utilities. Our ratio type evaluation inherently rules out the chance of bankruptcy from the admissible portfolio processes. Therefore, in contrast to \cite{TZ23}, our main results provide the portfolio strategies for the entire infinite horizon.

On the other hand, short-selling prohibition has attracted a lot of attention in academic studies, which has been widely used as a regulatory way to stabilize stock prices in the financial market and also as an effective method to minify the portfolio risk for the fund management. See, for instance, some empirical studies on the short-selling constraint and the economic influence in the financial market among \cite{AMP93}, \cite{CHS02}, \cite{ABCC04}, \cite{GS06}, \cite{BP13}, \cite{BJ13}, \cite{GM15}, and et al. See also some relevant studies on portfolio optimization with no short-selling such as the utility maximization (\cite{XS92(a)}, \cite{XS92(b)}, \cite{SH94}); the mean-variance portfolio optimization (\cite{LZ02}); the arbitrage theory in general market models (\cite{PS14}), and the references therein. 

The present paper aims to account for the short-selling constraint into the periodic evaluation of the relative performance. Inspired by \cite{TZ23}, we first reformulate the infinite horizon optimization problem with a periodic structure into an auxiliary one-period portfolio optimization problem under the short-selling constraint. To tackle the challenge of portfolio constraint, we employ the martingale duality approach developed in \cite{XS92(a)} and \cite{XS92(b)} and modify some technical arguments to fit into our formulation.  

We separate the discussions on the power utility and the logarithmic utility. The case of power utility is more interesting and technically involved as the optimal portfolio process differs fundamentally from its counterpart in the classical Merton problem. The heuristic characterization of the value function turns into a fixed point problem associated to a one-period optimization problem \eqref{A*.def} with a moderated utility function \eqref{2.8}. To tackle the auxiliary optimization problem with portfolio constraint and show the existence of the unique fixed point, we consider a proper dual control problem similar to \cite{XS92(a)} and \cite{XS92(b)}. Using some duality representation results and technical estimations, we can first prove the existence of the fixed point in the auxiliary problem. Then, we can take advantage of some results in the auxiliary one-period optimization problem and extend them in a periodic manner to obtain and verify the optimal constrained portfolio process in the original infinite horizon problem (see Theorem \ref{thm4.1}).

The case of logarithmic utility does not satisfy the sufficient conditions on utility functions in \cite{XS92(a)} and \cite{XS92(b)}, but we can similarly consider the dual control problem and manage to establish the duality result for the auxiliary optimization problem \eqref{A*.dual} within one period. More importantly, thanks to the property of the logarithmic utility, we can characterize the value function and the optimal feedback portfolio control explicitly for the original problem (see Theorem \ref{thm3.1}), where the feedback function actually is independent of the periodic structure and takes a unified form for the entire infinite horizon. It is also interesting to observe that under the logarithmic utility, the feedback function of the optimal constrained portfolio coincides with the result in the classical Merton problems under the short-selling constraint. In other words, under our new periodic evaluation formation with logarithmic utility, the optimal portfolio is to roll over the feedback control from the Merton solution forward in time.

Building upon the theoretical results, we also make some initial attempts in discussing the optimal choice of the period length tailor-made for the periodic evaluation. Indeed, contrary to the conventional utility maximization on terminal wealth, the effect of the time horizon vanishes in our framework. Instead, the value function and the optimal portfolio strategy evidently depend on the preset period length, especially in the case of power utility. Some other numerical illustrations on the sensitivity results with respect to model parameters are also presented.

The remainder of this paper is organized as follows. Section \ref{sec-formulation} introduces the market model and the problem formulation under the ratio type relative portfolio performance over an infinite horizon with the short-selling constraint. Section \ref{section4} investigates the case of power utility, where the characterization of the value function is equivalent to the existence of a fixed point to a nonlinear operator. The periodic characterization of the optimal portfolio is obtained therein together with some numerical illustrations. Section \ref{section3} further focuses on the optimization problem under logarithmic utility. With the help of the dual control problem and properties of the logarithmic function, the value function and the optimal portfolio process can be obtained explicitly over the entire horizon. Section \ref{sec-proof} collects all proofs of the results presented in the previous sections.

\section{Market Model and Problem Formulation}\label{sec-formulation}
Let $(\Omega,\mathcal{F},\{\mathcal{F}\}_{t\geq0},\mathbb{P})$ be a standard filtered probability space supporting a $n$-dimensional Brownian motion $W=(W^1_{t},...,W^n_{t})_{t\geq0}$. We consider a complete market model with $n$ risky assets and one risk-free asset. The risk-free asset has a constant interest rate $r\geq 0$ with the price process $B=(B_t)_{t\geq0}$. The price process of the $i$-th risky asset $S^i=(S^i_t)_{t\geq0}$ is governed by 
$$\frac{dS^i_t}{S^i_t}=\mu_i dt+\sum_{j=1}^n\sigma_{ij} dW^i_t,\quad t\geq 0.$$
We denote $\mu:=(\mu_1,...,\mu_n)^{\mathsf{T}}$ with $\mu_i\in\mathbb{R}$ as the mean return, and denote $\sigma:=(\sigma_{ij})$ with $\sigma_{ij}\in\mathbb{R}$ as the $n\times n$ covariance matrix. It is assumed that the strong non-degeneracy condition holds that
$$a^{\mathsf{T}}\sigma\sigma^{\mathsf{T}}a\geq \kappa_0\|a\|^2,\quad a\in\mathbb{R}^n,$$ for some $\kappa_0>0.$ Let us denote the Sharpe ratio
\begin{eqnarray}\label{xi}
\xi:=\sigma^{-1}{(\mu-r\mathbf{1})},
\end{eqnarray}
where $\mathbf{1}$ denotes the identity vector with each component equal to 1.

A trading strategy $\pi=(\pi^1_t,...,\pi^n_t)_{t\geq0}$ is a predictable process with $\pi^i_t$ representing the wealth invested in the $i$-th risky asset at time $t$. Under the given trading strategy $\pi$, we denote by $X=(X_t)_{t\geq0}$ the self-financing wealth process with the initial capital $X_0=x>0$ that satisfies the dynamics that 
\begin{eqnarray}
X_t^{\pi}=x+\sum_{i=1}^n\int_0^t{\pi^i_u}\frac{dS^i_u}{S^i_u}+\int_0^t(X^{\pi}_u-\mathbf{1}^{\mathsf{T}}\pi_u)\frac{dB_u}{B_u},\quad t\in[0,\infty).\nonumber
\end{eqnarray}

In the present paper, we are particularly interested in the portfolio constraint of no short-selling, i.e., it is required that $\pi^i_t\geq 0$ for $t\geq 0$, $i=1,\ldots,n$.

The dynamics of the controlled wealth $X^{\pi}$ (denoted by $X$ for short) can be rewritten as
\begin{eqnarray}
dX_t=
[rX_t+\pi_t^{\mathsf{T}}(\mu-r\mathbf{1})]dt+\pi_t^{\mathsf{T}}\sigma dW_t.\nonumber
\end{eqnarray}

Motivated by the recent study in \cite{TZ23}, we are interested in a new type of infinite horizon utility maximization problem on wealth $(X_t)_{t\geq0}$ when the evaluation of the wealth performance is conducted periodically at a sequence of prescribed future dates $(T_i)_{i\geq0}$ with $T_0:=0$ and $T_i\rightarrow\infty$ as $i\rightarrow\infty$. For simplicity, we shall consider the equal period $T_i=i\tau$ for $i\geq0$ with some $\tau>0$ such that the portfolio is evaluated every $\tau$ unit of time (e.g. monthly or annually). However, in contrast to the problem formulation in \cite{TZ23} where the portfolio performance is measured by the difference $X_{T_i}-\gamma X_{T_{i-1}}$ periodically between each time epoch $T_i$ and the previous date $T_{i-1}$, we consider the relative wealth performance in the present paper as periodic evaluation in the ratio type that 
\begin{eqnarray}
\frac{X_{T_i}}{(X_{T_{i-1}})^{\gamma}},\quad i\geq 1,\nonumber
\end{eqnarray}
for some relative performance parameter $\gamma\in(0,1]$. That is, when $\gamma\rightarrow 1$, the agent cares more about the relative ratio comparing with the wealth level at the previous evaluation date $T_{i-1}$; while $\gamma\rightarrow 0$ indicates that the agent values more on the absolute value of the current wealth.

As a direct consequence, while the optimal wealth $X_t$ is allowed to hit $0$, i.e., the bankruptcy $X_t=0$ may occur at some time $t$ and will remain at the $0$ wealth level afterwards in \cite{TZ23}, our formulation in the ratio type enforces the positive wealth constraint $X_t>0$ as an inherent nature by the problem formulation and the resulting optimal wealth will retain positive for all time $t\geq 0$, which makes more sense if we are interested in the periodic evaluation in the infinite horizon manner.

Mathematically speaking, we aim to solve the optimal portfolio problem by periodically evaluating the ratio type of relative performance in the infinite horizon fashion, and our optimization problem is formulate as
\begin{align}\label{utilitymax}
\sup_{\pi\in\mathcal{U}_0(x)}\ \ \sum_{i=1}^{\infty}e^{-\delta T_i}\mathbb{E}\left[U\left(\frac{X_{T_i}}{\left(X_{T_{i-1}}\right)^{\gamma}}\right)\right],
\end{align}
where $\delta>0$ is the agent's subjective discount factor, and for technical convenience, we shall only consider two types of utility functions in the present paper, namely the power utility function $U(x)=\frac{1}{\alpha} x^{\alpha}$ with $\alpha\in(-\infty,0)\cup(0,1)$ and the logarithmic utility function $U(x)=\log x$. 

Moreover, the set of admissible portfolio processes in problem \eqref{utilitymax} is defined by 
\begin{eqnarray}\label{set.Ut}
\mathcal{U}_0(x)
\hspace{-0.3cm}&:=&\hspace{-0.3cm}
\bigg\{\pi:\pi=(\pi^1_t,...,\pi^n_t)_{t\geq0}\text{ is a predictable and locally square-integral process such that}
\nonumber\\
\hspace{-0.3cm}&&\hspace{-0.3cm}
\pi^i_t\geq0\text{ for }1\leq i\leq n,\text{ }X_t=x+\sum_{i=1}^n\int_0^t{\pi^i_u}\frac{dS^i_u}{S^i_u}+\int_0^t(X_u-\mathbf{1}^{\mathsf{T}}\pi_u)\frac{dB_u}{B_u}>0 \text{ for } t\geq 0,
\nonumber\\
\hspace{-0.3cm}&&\hspace{-0.3cm}
\color{black}\text{and } 
\sum_{i=1}^{\infty}e^{-\delta T_i}\mathbb{E}\left[\left(U\left(\frac{X_{T_i}}{\left(X_{T_{i-1}}\right)^{\gamma}}\right)\right)_{-}\right]<\infty
\bigg\},\quad x\in\mathbb{R}_+,
\end{eqnarray}
with $a_{-}=\max\{-a,0\}$. As stated earlier, the condition that $\pi^i_t\geq 0$ rules out the short-selling strategy in the $i$-th risky asset. On the other hand, $X_t-\mathbf{1}^{\mathsf{T}}\pi_t$ is allowed to be negative, i.e., borrowing from the money market is permitted in the present paper. We also note that the last condition is of the integrability type, which is needed to ensure the well-posedness of the periodic portfolio evaluation and optimization problems in Section \ref{section4} and \ref{section3}.

\section{Periodic Evaluation under Power Utility}\label{section4}
In this section, we first consider the case of power utility function
$$U(x):=\frac{1}{\alpha}x^{\alpha},\quad x\in\mathbb{R}_+,$$ with $\alpha\in (-\infty, 0)\cup(0,1)$ respresenting the risk aversion of the agent. The value function of the optimal portfolio under periodic evaluation is written by
\begin{eqnarray}\label{problem}
V(x):=\sup_{\pi\in\mathcal{U}_0(x)}\mathbb{E}\left[\sum_{i=1}^{\infty}e^{-\delta T_i}\frac{1}{\alpha}\left(\frac{X_{T_i}}{\left(X_{T_{i-1}}\right)^{\gamma}}\right)^{\alpha}\right],
\end{eqnarray}
where $\mathcal{U}_0$ is the set of admissible portfolio processes defined by \eqref{set.Ut}. Again, using the dynamic programming principle, we can heuristically derive that
\begin{eqnarray}\label{ddp}
V(x)=\sup_{\pi\in\mathcal{U}_0(x)}\mathbb{E}\left[e^{-\delta T_1}\frac{1}{\alpha}\frac{X^{\alpha}_{T_1}}{x^{\alpha\gamma}}+e^{-\delta T_1}V(X_{T_1})\right].
\end{eqnarray}

For the wellposedness of the problem, the following standing assumption is imposed throughout this section.
\begin{Assumption}
\label{ass1}
The model parameters satisfy that
\begin{eqnarray}
\delta>\zeta(\alpha(1-\gamma))\vee0,\nonumber
\end{eqnarray}
where the function $$\zeta(x):=rx+\frac{x\|\Tilde{\xi}\|^2}{2(1-x)},\quad x\in(-\infty,1),$$
and $\Tilde{\xi}=\xi+\sigma^{-1}\Tilde{\pi}^*$, with $\Tilde{\pi}^*$ being the minimizer of $\|\xi+\sigma^{-1}\Tilde{\pi}\|$ in $\Tilde{\pi}\in[0,\infty)^n.$
\end{Assumption}

In view of the scaling property of the utility function $U(x)=x^{\alpha}U(1)$, we heuristically conjecture our value function in the form of $V(x)=\frac{1}{\alpha}A^*x^{\alpha(1-\gamma)}$ for some non-negative constant $A^*$. Substituting this expression of $V$ in \eqref{ddp} and dividing both sides by $\frac{1}{\alpha}x^{\alpha(1-\gamma)}$, we obtain that
\begin{eqnarray}\label{A*.def}
A^*
\hspace{-0.3cm}&=&\hspace{-0.3cm}
{\alpha}\sup_{\pi\in\mathcal{U}_0(x)}\mathbb{E}\left[e^{-\delta T_1}U\left(\frac{X_{T_1}}{x}\right)+e^{-\delta T_1}\frac{1}{\alpha}A^*\left(\frac{X_{T_1}}{x}\right)^{\alpha(1-\gamma)}\right]
\nonumber\\
\hspace{-0.3cm}&=&\hspace{-0.3cm}
{\alpha}\sup_{\pi\in\mathcal{U}_0(1)}\mathbb{E}\left[e^{-\delta \tau}U\left(X_{\tau}\right)+e^{-\delta \tau}\frac{1}{\alpha}A^*X_{\tau}^{\alpha(1-\gamma)}\right].
\end{eqnarray}
Consequently, in the case of power utility, the characterization of the value function $V$ can be equivalently formulated as the characterization of the unknown non-negative constant $A^*$ as a fixed-point of the operator defined in a one-dimensional Euclidean space. To this end, we will first establish the existence of the optimizer to the auxiliary problem \eqref{A*.def} with a given $A^*$, and then show the existence of the fixed point $A^*$ to the operator. 

To cope with the short-selling constraint,
we follow \cite{XS92(a)} to employ the martingale duality approach.
To this purpose, let us introduce $h_a:\mathbb{R}_+\rightarrow\mathbb{R}$ by
\begin{eqnarray}
\label{2.8}
h_a(x):=\frac{1}{\alpha}x^{\alpha}+\frac{1}{\alpha}ax^{\alpha(1-\gamma)},
\end{eqnarray}
where $a\in\mathbb{R}_+$ is regarded as a parameter of the function $h_a$.

As a preparation for the main result, we first derive some preliminary  properties of the function $h_a$, which plays a pivotal role in later proofs.
\begin{lem}\label{lem2.1}
For a fixed $a\in\mathbb{R}_+$, the function $h_a(x)$ defined by \eqref{2.8} is strictly increasing and strictly concave on $(0,\infty)$, and it holds that $h_a^{\prime}(0+)=\infty$ and $h_a^{\prime}(\infty)=0$. In addition, there exist some constants $\vartheta\in(0,1)$ and $\varrho\in(1,\infty)$ such that 
\begin{eqnarray}\label{h_a'>h_a'}
\vartheta h_a^{\prime}(x)\geq h_a^{\prime}(\varrho x),\quad x\in\mathbb{R}_+.
\end{eqnarray} 
Furthermore, when $\alpha\in(0,1)$, there exist some constants $\kappa_1\in(0,\infty)$ and $\rho_1\in(0,1)$ such that
\begin{eqnarray}\label{h_a<kappa}
0< h_a(x)\leq \kappa_1(1+x^{\rho_1}),\quad x\in\mathbb{R}_+.
\end{eqnarray}
\end{lem}

We can now rewrite the fixed point problem \eqref{A*.def} equivalently as the fixed point problem by
\begin{eqnarray}\label{problem3}
A^*=e^{-\delta\tau}H(A^*),
\end{eqnarray}
where the function $\mathbb{R}_+\ni a\mapsto H(a)\in \mathbb{R}_+$ is defined by
\begin{eqnarray}\label{poweraux}
H(a):=\alpha\sup_{\pi\in\mathcal{U}_{0}(1)}\mathbb{E}[h_a(X_{\tau})]=\alpha\sup_{\pi\in\mathcal{U}_{0}(1)}\mathbb{E}\left[\frac{1}{\alpha}X_{\tau}^{\alpha}+\frac{1}{\alpha}a X_{\tau}^{\alpha(1-\gamma)}\right].
\end{eqnarray}

Note that, for any fixed $a$, the auxiliary problem in \eqref{poweraux} is a one-period portfolio optimization problem under short-selling constraint, however, with the moderated utility function $h_a(x)$ instead of the standard power utility.


\subsection{Duality Method and Main Results}
We will first formulate a dual problem and establish the existence of the dual optimizer. Then, we close the duality gap between the auxiliary primal problem \eqref{poweraux} and the dual problem and characterize the primal optimizer satisfying the short-selling constraint using the dual optimal solution.


Let us first introduce some notations as follows. Recalling the Sharpe ratio $\xi$ in \eqref{xi}, let us define the martingale state price density process $Z=(Z_t)_{t\in[0,\tau]}$ by
\begin{eqnarray}\label{z.process}
Z_t=\exp\left(-{\xi}^{\mathsf{T}}W_t-\frac{1}{2}\|\xi\|^2 t\right),\quad t\in[0,\tau].
\end{eqnarray}
For a given function $f$, let us denote its Legendre-Fenchel transform by
\begin{eqnarray}
\Phi_f(y):=\sup_{x\geq0}(f(x)-yx),\quad y\in\mathbb{R}_+.\nonumber
\end{eqnarray}
Provided that $f$ is continuous and concave with $f^{\prime}(\infty)=0$, the maximizer attaining the supremum always exists (although not necessarily unique), which is denoted by $x_f^*(y):=\arg\max_{x\geq0}(f(x)-yx)$. It follows that $\Phi_f(y)=f(x_f^*(y))-yx_f^*(y)$.

By Lemma \ref{lem2.1}, one can define the inverse function of $h^{\prime}_a$ by $\mathbb{R}_+\ni y\mapsto I(y)\in\mathbb{R}_+$, which is a strictly decreasing function.
Let 
$\Phi_{h_a}(y)=\sup_{x\geq0}\{h_a(x)-xy\}$ be the Legendre-Fenchel transform of the concave function $h_a$. It is well known that
\begin{eqnarray}\label{phi(y)=h}
h_a(x^*_{h_a}(y))-x^*_{h_a}(y)y=h_a(I(y))-I(y)y,
\quad y\in\mathbb{R}_+,
\end{eqnarray}
as well as when $\alpha\in(0,1)\, (\alpha\in(-\infty,0),\,resp.)$
\begin{eqnarray}
\label{phi(0)}
\Phi_{h_a}(0)
\hspace{-0.3cm}&:=&\hspace{-0.3cm}
\lim_{y\rightarrow0+}\Phi_{h_a}(y)=h_a(\infty)=\infty\, (0,\,resp.),\\
\label{phi(infty)}
\Phi_{h_a}(\infty)
\hspace{-0.3cm}&:=&\hspace{-0.3cm}
\lim_{y\rightarrow\infty}\Phi_{h_a}(y)=h_a(0)=0\, (-\infty,\,resp.),
\end{eqnarray}
and 
\begin{eqnarray}\label{Phi'}
\Phi_{h_a}^{\prime}(y)=-x^*_{h_a}(y),\quad \Phi_{h_a}^{\prime\prime}(y)=-x_{h_a}^{*\,\prime}(y),\quad y\in\mathbb{R}_+.
\end{eqnarray}
In particular, the function $\Phi_{h_a}$ is a strictly decreasing, strictly convex and twice differentiable function. 
Define the set of dual control processes by
\begin{eqnarray}\label{tilde.U}
\Tilde{\mathcal{U}}_{s,t}
\hspace{-0.3cm}&:=&\hspace{-0.3cm}
\left\{\Tilde{\pi}:\Tilde{\pi}=(\Tilde{\pi}^1_u,...,\Tilde{\pi}^n_u)^{\mathsf{T}}_{u\in[s,t]}\text{ is a predictable process satisfying}\right.
\nonumber\\
\hspace{-0.3cm}&&\hspace{0.3cm}
\left.\text{$\Tilde{\pi}^i_u\geq0$ for $1\leq i\leq n$, and
$\mathbb{E}\left[\int_s^t\|\Tilde{\pi}_u\|^2du\right]<\infty$}\right\}.
\end{eqnarray}
For each $\Tilde{\pi}\in\Tilde{\mathcal{U}}_{0,\tau}$, let us consider the positive local martingale
\begin{eqnarray}\label{Z.pi}
Z^{\Tilde{\pi}}_t:=\exp\left(-\int_0^t\left[\xi+{\sigma}^{-1}{\Tilde{\pi}_s}\right]^{\mathsf{T}}dW_s-\frac{1}{2}\int_0^t\left\|\xi+{\sigma}^{-1}{\Tilde{\pi}_s}\right\|^2ds\right),\quad t\in[0,\tau].
\end{eqnarray}

Similar to \cite{XS92(a)} and \cite{XS92(b)}, we consider the objective of the dual control problem associated to the primal problem \eqref{problem3} that
\begin{eqnarray}
\Tilde{J}_a(y,\Tilde{\pi}):=\mathbb{E}\left[\Phi_{h_a}(y\frac{Z^{\Tilde{\pi}}_{\tau}}{B_{\tau}})\right]=\mathbb{E}\left[h_a(x^*_{h_a}(y\frac{Z^{\Tilde{\pi}}_{\tau}}{B_{\tau}}))-x^*_{h_a}(y\frac{Z^{\Tilde{\pi}}_{\tau}}{B_{\tau}})y\frac{Z^{\Tilde{\pi}}_{\tau}}{B_{\tau}}\right].\nonumber
\end{eqnarray}
The dual value function $\Tilde{V}_a$ is defined by
\begin{eqnarray}\label{dual.pro}
\Tilde{V}_a(y)=\inf_{\Tilde{\pi}\in\Tilde{\mathcal{U}}_{0,\tau}}\Tilde{J}_a(y,\Tilde{\pi}),\quad y\in\mathbb{R}_+.
\end{eqnarray}

 The existence of the dual optimizer to problem \eqref{dual.pro} is guaranteed in the next result.

\begin{prop}\label{existence}
The unique constant vector minimizer of $\|\xi+\sigma^{-1}\Tilde{\pi}\|$ over $\Tilde{\pi}\in[0,\infty)^n$ exists and is denoted by $\Tilde{\pi}^*\in\Tilde{\mathcal{U}}_{0,\tau}$, which is also
the optimal solution to the dual problem \eqref{dual.pro}. More specifically, with
\begin{align}
    \label{tilde.xi}\Tilde{\xi}&:=\xi+{\sigma}^{-1}{\Tilde{\pi}^*},\\
    \label{Z.tilde}Z^{\Tilde{\pi}^*}_t&:=\exp\left\{-\Tilde{\xi}^{\mathsf{T}}W_t-\frac{1}{2}\|\Tilde{\xi}\|^2t\right\},\quad t\in[0,\tau],
\end{align}
the value function of dual problem \eqref{dual.pro} is given by
\begin{eqnarray}
\Tilde{V}_a(y)=\mathbb{E}\left[\Phi_{h_a}(y\frac{Z^{\Tilde{\pi}^*}_{\tau}}{B_{\tau}})\right],\quad y\in\mathbb{R}_+.\nonumber
\end{eqnarray}
In addition, for any fixed $a\in\mathbb{R}_+$, the dual function $\Tilde{V}_a$ is finite, 
continuous, non-increasing and convex on $\mathbb{R}_+$, and it holds that when $\alpha\in(0,1)\,(\alpha\in(-\infty,0),\,resp.)$
\begin{eqnarray}
\Tilde{V}_a(0)=\lim_{y\rightarrow0+}\mathbb{E}\left[\Phi_{h_a}(y\frac{Z^{\Tilde{\pi}^*}_{\tau}}{B_{\tau}})\right]=\infty\,(0,\,resp.),\quad \Tilde{V}_a(\infty)=\lim_{y\rightarrow\infty}\mathbb{E}\left[\Phi_{h_a}(y\frac{Z^{\Tilde{\pi}^*}_{\tau}}{B_{\tau}})\right]=0\,(-\infty,\,resp.).\nonumber
\end{eqnarray}  
\end{prop}


Recall that the constant $\Tilde{\pi}^*$ is the optimal control to the dual problem \eqref{dual.pro} for all fixed $a,y\in\mathbb{R}_+$. Let us define a function $\mathbb{R}_+\ni \lambda\mapsto\ell_{a,y}(\lambda)\in\mathbb{R}_+$ by
\begin{eqnarray}\label{ell}
\ell_{a,y}(\lambda):=\Tilde{J}_a(\lambda y,\Tilde{\pi}^*)=\mathbb{E}\left[\Phi_{h_a}(\lambda y\frac{Z^{\Tilde{\pi}^*}_{\tau}}{B_{\tau}})\right]=\mathbb{E}\left[h_a(x^*_{h_a}(\lambda y\frac{Z^{\Tilde{\pi}^*}_{\tau}}{B_{\tau}}))-x^*_{h_a}(\lambda y\frac{Z^{\Tilde{\pi}^*}_{\tau}}{B_{\tau}})\lambda y\frac{Z^{\Tilde{\pi}^*}_{\tau}}{B_{\tau}}\right],
\end{eqnarray}
where, in the last equality, we have used \eqref{phi(y)=h}. We have the next preparation result to characterize the minimizer of $\inf_{y\in\mathbb{R}_+}\{\Tilde{V}_a(y)+xy\}$.

\begin{prop}\label{unique}
For any fixed $a,x\in\mathbb{R}_+$, there exists a unique $y_x^*\in\mathbb{R}_+$ such that
\begin{eqnarray}\label{H(a)=inf}
\inf_{y\in\mathbb{R}_+}\{\Tilde{V}_a(y)+xy\}=\Tilde{V}_a(y_x^*)+xy_x^*.
\end{eqnarray}
Furthermore, $y_x^*\in\mathbb{R}_+$ is the minimizer of $$\inf_{y\in\mathbb{R}_+}\{\Tilde{V}_a(y)+xy\},$$ if and only if $y_x^*\in\mathbb{R}_+$ satisfies
\begin{eqnarray}
\label{ell(1).y*}
\ell^{\prime}_{a,y^*_x}(1)+xy_x^*=0.
\end{eqnarray}
\end{prop}

The next result establishes the relationship between the optimal solutions to the primal problem \eqref{problem3} and dual problem \eqref{dual.pro}. In addition, we can  characterize an optimal portfolio process $\pi^*\in{\mathcal{U}}_{0,\tau}$ with no-short selling and the resulting optimal wealth process $X^*$. 

\begin{prop}\label{relationship}
For any fixed $a\in\mathbb{R}_+$, the value function of the primal problem \eqref{problem3} satisfies
\begin{eqnarray}
H(a)=\alpha\sup_{\pi\in\mathcal{U}_{0}(1)}\mathbb{E}[h_a(X_{\tau})]={\alpha}\inf_{y\in\mathbb{R}_+}\{\Tilde{V}_a(y)+y\}
={\alpha}(\Tilde{V}_a(y^*)+y^*),\nonumber
\end{eqnarray}
where $y^*=\left.y^*_{x}\right|_{x=1}$ is characterized in Proposition \ref{unique}.
The optimal wealth process $X^*$ is given by $X^*_{\tau}=x^*_{h_a}\left(y^*\frac{Z^{\Tilde{\pi}^*}_{\tau}}{B_{\tau}}\right)$.
Moreover, denote by $(\eta_t)_{t\in[0,\tau]}$ the $\mathbb{R}^n$-valued and $\{\mathcal{F}_t\}$-progressively measurable process that is determined by \eqref{mart.repre.theo.} such that $\int_0^{\tau}\|\eta_t\|^2dt<\infty$. Then, the process $X^*$ admits a representation that
\begin{eqnarray}\label{hat.X}
\hat{X}_t=\frac{B_{t}}{Z^{\Tilde{\pi}^*}_{t}}\left(1+\int_0^t\eta^{\mathsf{T}}_sdW_s\right),\quad t\in[0,\tau].
\end{eqnarray}
The optimal portfolio process is given by
\begin{eqnarray}
\pi^*_t=\hat{X}_t(\sigma^{\mathsf{T}})^{-1}\left[\xi+{\sigma}^{-1}{\Tilde{\pi}_t^*}+\eta_t\left(1+\int_0^t\eta^{\mathsf{T}}_sdW_s\right)^{-1}\right],\quad t\in[0,\tau].\nonumber
\end{eqnarray}
\end{prop}

With the help of the relationship between the optimal solutions of the primal problem \eqref{problem3} and dual problem \eqref{dual.pro} in Proposition \ref{relationship}, one can obtain that the function $H(\cdot)$ given by \eqref{problem3} is well-defined. In what follows, we show that there indeed exists a unique fixed point $A^*\in\mathbb{R}_+$ solving
equation \eqref{problem3}.

\begin{prop}\label{fixppower}
Define a function $\Psi: \mathbb{R}_{+}\rightarrow\mathbb{R}_{+}$ as
\begin{eqnarray}\label{Psi}
\Psi(a):=e^{-\delta \tau}H(a),
\end{eqnarray}
where $H(a)$ is the solution to problem \eqref{problem3}. Then $\Psi$ is a contraction mapping on the metric space $(\mathbb{R}_{+},|\cdot|)$, where $|\cdot|$ is the absolute value norm. Consequently, $\Psi$ admits a unique fixed-point $A^*$ such that $A^*=\Psi(A^*)$. 
Moreover, the unique fixed-point $A^*$ of $\Psi$ satisfies
\begin{eqnarray}
\frac{e^{(r\alpha-\delta)\tau}}{1-e^{-(\delta-r\alpha(1-\gamma))\tau}}\mathbf{1}_{\{\alpha\in(0,1)\}}\hspace{-0.3cm}&+&\hspace{-0.3cm}\frac{e^{(\zeta(\alpha)-\delta)\tau}}{1-e^{-(\delta-\zeta(\alpha(1-\gamma)))\tau}}\mathbf{1}_{\{\alpha<0\}}
\nonumber\\
\hspace{-0.3cm}&&\hspace{-0.3cm}
 \leq A^*
\leq\frac{e^{(\zeta(\alpha)-\delta)\tau}}{1-e^{-(\delta-\zeta(\alpha(1-\gamma)))\tau}}\mathbf{1}_{\{\alpha\in(0,1)\}}+\frac{e^{(r\alpha-\delta)\tau}}{1-e^{-(\delta-r\alpha(1-\gamma))\tau}}\mathbf{1}_{\{\alpha<0\}}.\nonumber
\end{eqnarray}
\end{prop}

With the previous preparations, 
we are ready to prove the verification theorem for the portfolio optimization problem \eqref{problem} under periodic evaluation in the ratio type.

\begin{thm}[Verification Theorem]\label{thm4.1}
Let $y^*$ be the unique solution to $\ell^{\prime}_{A^*,y^*}(1)+y^*=0.$
The value function of the infinite horizon original problem \eqref{problem} is given by
\begin{eqnarray}
V(x)=\frac{1}{\alpha}A^*x^{\alpha(1-\gamma)},\quad x\in\mathbb{R}_+,\nonumber
\end{eqnarray}
where $A^*\in\mathbb{R}_+$is the unique fixed-point of the function $\Psi$ defined in \eqref{Psi}. In particular, the optimal wealth process $X^*$ at time $T_i$ is given by
\begin{eqnarray}
X^*_{T_i}=X^*_{T_{i-1}}x^*_{h_{A^*}}\left(y^*\frac{Z^{\Tilde{\pi}^*}_{T_i}/B_{T_{i}}}{Z^{\Tilde{\pi}^*}_{T_{i-1}}/B_{T_{i-1}}}\right),\quad i\geq 1,\nonumber
\end{eqnarray}
with $X^*_{T_0}=x$, where the function $x^*_{h_{A^*}}(\cdot)$ is the inverse function of $h^{\prime}_{A^*}$ and the process $Z^{\Tilde{\pi}^*}$ is given by \eqref{Z.tilde}.
Moreover, the process $X^*$ has a representation that
\begin{eqnarray}\label{X.rep}
X^*_t=X^*_{T_i}\frac{Z^{\Tilde{\pi}^*}_{t}/B_{t}}{Z^{\Tilde{\pi}^*}_{T_i}/B_{T_i}}\left(1+\int_{T_i}^{t}\eta^{\mathsf{T}}_sdW_s\right),\quad t\in[T_i,T_{i+1}],
\end{eqnarray}
and the optimal portfolio process can be expressed by
\begin{eqnarray}\label{pi.rep}
\pi^*_t=X^*_t(\sigma^{\mathsf{T}})^{-1}\left[\xi+{\sigma}^{-1}{\Tilde{\pi}^*}+\eta_t\left(1+\int_{T_i}^{t}\eta^{\mathsf{T}}_sdW_s\right)^{-1}\right],\quad t\in[T_i,T_{i+1}].
\end{eqnarray}
\end{thm}

\begin{corollary}\label{coro4.1}
If $\gamma=1$, then the value function of \eqref{problem} is simplified to
\begin{eqnarray}
V(x)=\frac{e^{(\zeta(\alpha)-\delta)\tau}}{\alpha(1-e^{-\delta\tau})},\quad x\in\mathbb{R}_+.\nonumber
\end{eqnarray}
In particular, the optimal wealth process $X^*$ at time $T_i$ is simplified to
\begin{eqnarray}
X^*_{T_i}=X^*_{T_{i-1}}e^{\frac{\zeta(\alpha)}{\alpha-1}\tau}\left(\frac{Z^{\Tilde{\pi}^*}_{T_i}/B_{T_i}}{Z^{\Tilde{\pi}^*}_{T_{i-1}}/B_{T_{i-1}}}\right)^{\frac{1}{\alpha-1}},\quad t\in[0,\infty),\nonumber
\end{eqnarray}
and the process $X^*$ has a representation that
\begin{eqnarray}
X^*_{t}=
xe^{\frac{\zeta(\alpha)}{\alpha-1}t}\left(\frac{Z^{\Tilde{\pi}^*}_{t}}{B_{t}}\right)^{\frac{1}{\alpha-1}},\quad t\in[0,\infty).\nonumber
\end{eqnarray}
Moreover, the optimal portfolio process admits the feedback control form that
\begin{align*}
&\pi^*_t=\pi^*(X_t^*),\\
\text{where}\ \ &\pi^*(x)=\frac{x}{1-\alpha}(\sigma^{\mathsf{T}})^{-1}[\xi+{\sigma}^{-1}{\Tilde{\pi}^*}],\quad t\in[0,\infty).
\end{align*}
\end{corollary}

\subsection{Quantitative Properties and Numerical illustrations}

Under the power utility with $\gamma<1$, the optimal portfolio process and the optimal wealth have no unified expressions over the entire infinite horizon. Indeed, we can only derive their implicit expressions within each period. We can see that the optimal portfolio \eqref{pi.rep} to the problem \eqref{problem} now exhibits the periodic feature stemming from the problem formulation with periodic evaluation. Only in the extreme case when $\gamma=1$, the optimal portfolio control can be expressed as a feedback form over the entire infinite horizon.

In what follows, we would like to numerically illustrate some quantitative properties of the optimal wealth process and the optimal portfolio control in Theorem \ref{thm4.1}. For simplicity, we only consider the market model with two stocks. Let
the mean rates of return $\mu=(\mu_1,\mu_2)^{\mathsf{T}}$ and the $2\times2$ volatility matrix $\sigma$ be given by
\begin{eqnarray}
    \mu=\left(
    \begin{array}{cc}
         0.1 \\
         0.15
    \end{array}\right),\quad
    \sigma=\left(
    \begin{array}{cc}
        0.2 & 0 \\
        0 & 0.25
    \end{array}
    \right),\nonumber
\end{eqnarray}
and let other market parameters be given as below (unless otherwise specified):
\begin{table}[H]
\centering
\caption{The values of the parameters.}
\label{tab2}
 \begin{tabular}{lcccccc}
\hline\hline
\noalign{\smallskip}
Parameter & $r$ & $\tau$ & $\gamma$ & $\alpha$ & $\delta$ &$x$ \\\noalign{\smallskip}
Value & 0.12& 1 & 0.8 & 0.5&  0.3 &0.5 \\\noalign{\smallskip}
\hline
\end{tabular}   
\end{table}

Recalling the value function $V(x)={\frac{1}{\alpha}}A^*x^{\alpha(1-\gamma)}$ with $A^*$ being the fixed-point of the operator $\Psi$ given in \eqref{Psi}, we first illustrate the sensitivity of $A^*$ (and hence of $V(x)$) in top panels of Figure \ref{fig-power} with respect to the mean return $\mu$ and the volatility $\sigma$ of stock 2 while keeping parameters of stock 1 fixed. It is shown that when the mean return $\mu_2$ is low, the short-selling constraint will force zero investment in stock 2. When the mean return $\mu_2$ is sufficiently large such that the agent invests the positive wealth in the stock, the value function is increasing with respect to the mean return $\mu_2$ and is decreasing with respect to the variance $\sigma_{22}$. We also plot the sensitivity of the value function $V(x)$ with respect to the risk aversion coefficient $\alpha$. We can see from the bottom panel of Figure \ref{fig-power} that the value function is essentially decreasing in the risk aversion parameter $\alpha$. In particular, for $\alpha\in (0,1)$, the value function is decreasing in $\alpha$ from infinity to zero; and for $\alpha<0$, the value function is decreasing from zero to negative infinity, which coincide with the monotonicity results in $\alpha$ observed in the classical Merton solution. 

\begin{figure}[H]
    \centering
\includegraphics[width=6.5cm]{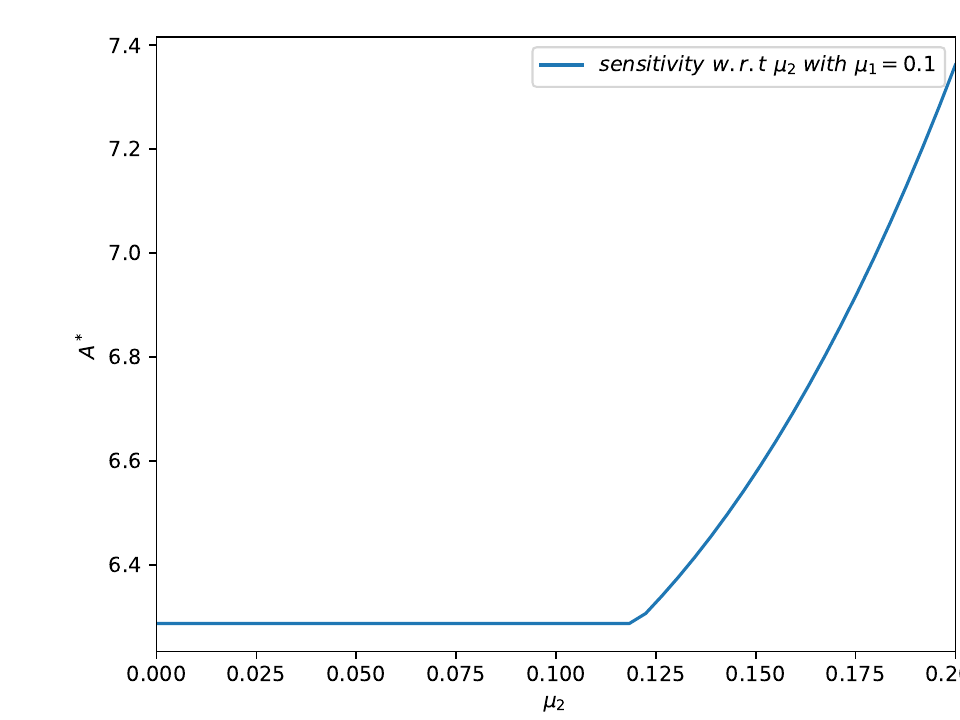}
   \quad
   \includegraphics[width=6.5cm]{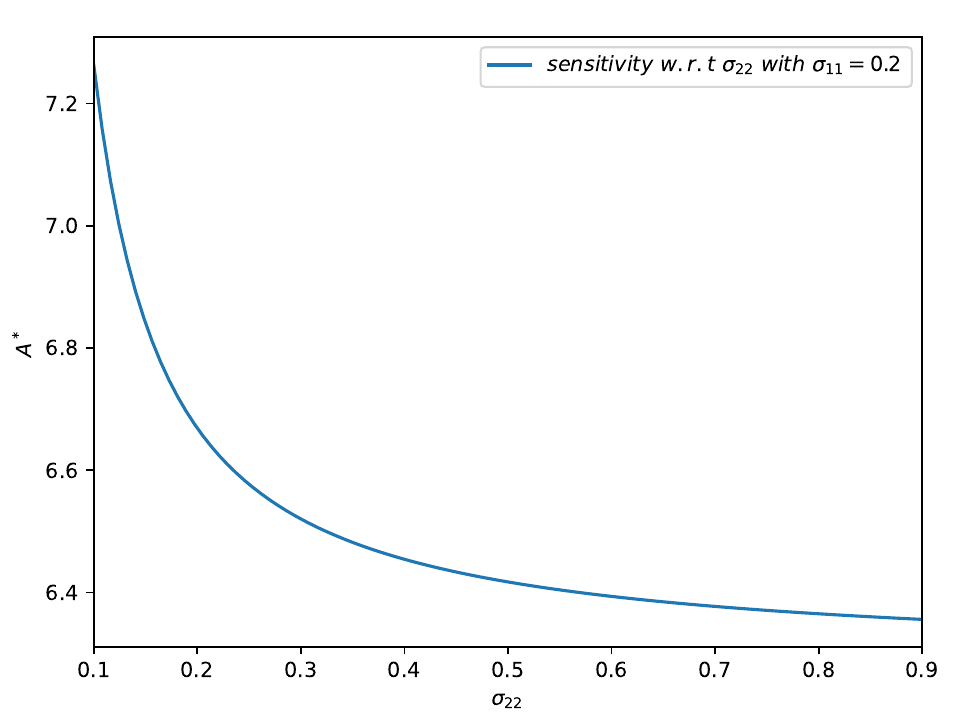}\\
   \quad
\includegraphics[width=6.5cm]{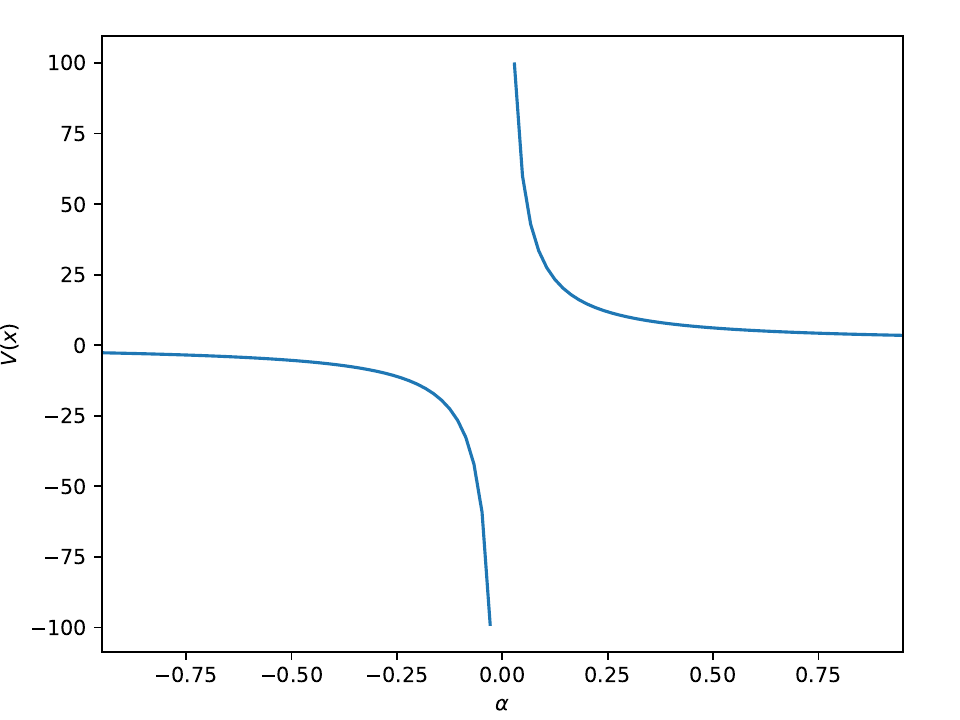}
    \caption{The sensitivity of $A^*$ with respect to the expected return $\mu$ (left-top panel) and the variance $\sigma$ (right-top panel); the sensitivity of $V(x)$ with respect to the risk aversion coefficient $\alpha$ (bottom panel)}\label{fig-power}
\end{figure}

We also plot the sensitivity of $V(x)$ with respect to the relative performance parameter $\gamma\in(0,1]$ in Figure \ref{fig-power2}. Different trends can be observed sensitively depending on the initial wealth $x$ and the risk aversion parameter $\alpha$. For the agent with a high risk aversion parameter $\alpha<0$, the value function is decreasing in $\gamma$ when the initial wealth is low; and the trend is reversed when the initial wealth is high. On the other hand, for the agent who is less risk averse, the value function is increasing in the parameter $\gamma$ when the initial wealth is high and it is decreasing in $\gamma$ with low initial wealth.

\begin{figure}[H]
    \centering
\includegraphics[width=6.5cm]{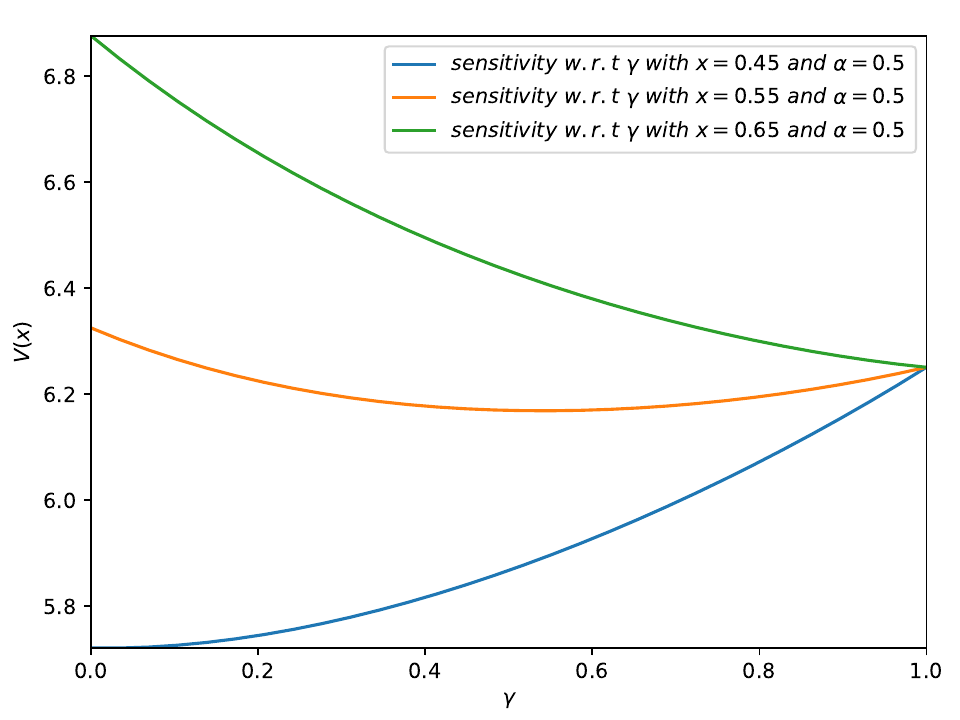}
\quad
\includegraphics[width=6.5cm]{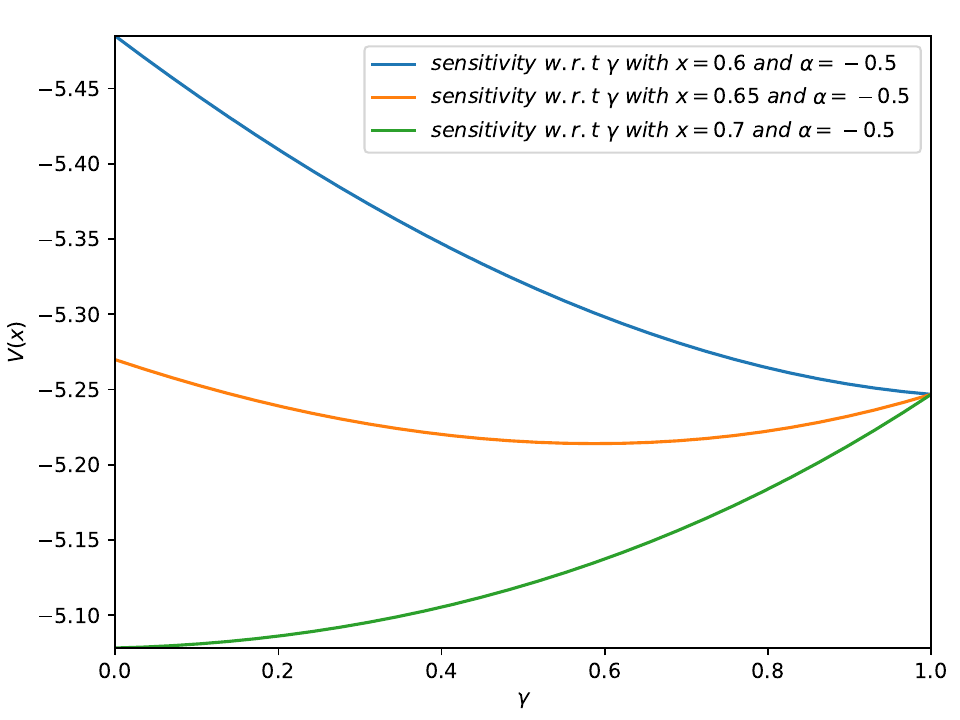}

\caption{The sensitivity $V(x)$ with respect to the relative performance parameter $\gamma$ }\label{fig-power2}

\end{figure}

\ \\

By the nature of the periodic evaluation over an infinite horizon, the conventional dependence of the value function on the time horizon $T$ is eliminated in our problem formulation. Instead, from the definition of $V(x)$, the choice of the evaluation period length $\tau$ may play an important role in determining the optimal portfolio and the value function. However, under power utility, the constant $A^*$ in the value function $V(x)={\frac{1}{\alpha}}A^*x^{\alpha(1-\gamma)}$ is an implicit fixed point to the operator $\Psi$ in \eqref{Psi}, which makes the theoretical dependence on the period length $\tau$ obscure in general. To highlight the dependence of the value function $V(x)$ on the choice of the period length, let us rewrite $V(x)$ as $V(x;\tau)$. In fact, we can numerically show that the value function itself may not admit the optimal choice of $\tau^*$ under power utility function, making the optimal periodicity in this case not well defined. However, from the perspective to scale the frequency in the periodic evaluation, it is also reasonable to consider the target as the scaled value function $V(x;\tau)\tau$. In Figure \ref{fig-power-tau}, we can illustrate the sensitivity results of $A^*\tau$ (and hence the scaled value function $V(x;\tau)\tau$) with respect to the choice of $\tau$ when $\gamma=0.8$ and $\gamma=1$, respectively. We can see that there exists a unique optimal periodicity of $\tau^*$ in each plot that the scaled value function can be optimized. We stress that the choice of $\tau$ not only affects the value function $V(x;\tau)$, but also has the direct and significant impact on the optimal wealth $X^*_t$ in \eqref{X.rep} and the optimal portfolio process in \eqref{pi.rep}. Different period length $\tau$ eventually leads to very distinct decision making in the infinite horizon portfolio management. In addition, more frequent evaluations in reality may also cause higher costs. Therefore, it will be appealing to formulate and incorporate some proper cost functions on the period length $\tau$ and study the optimization problem over $\tau$ in the context of portfolio periodic evaluation, which will be left as our future research.

\begin{figure}[H]
    \centering
\includegraphics[width=6.5cm]{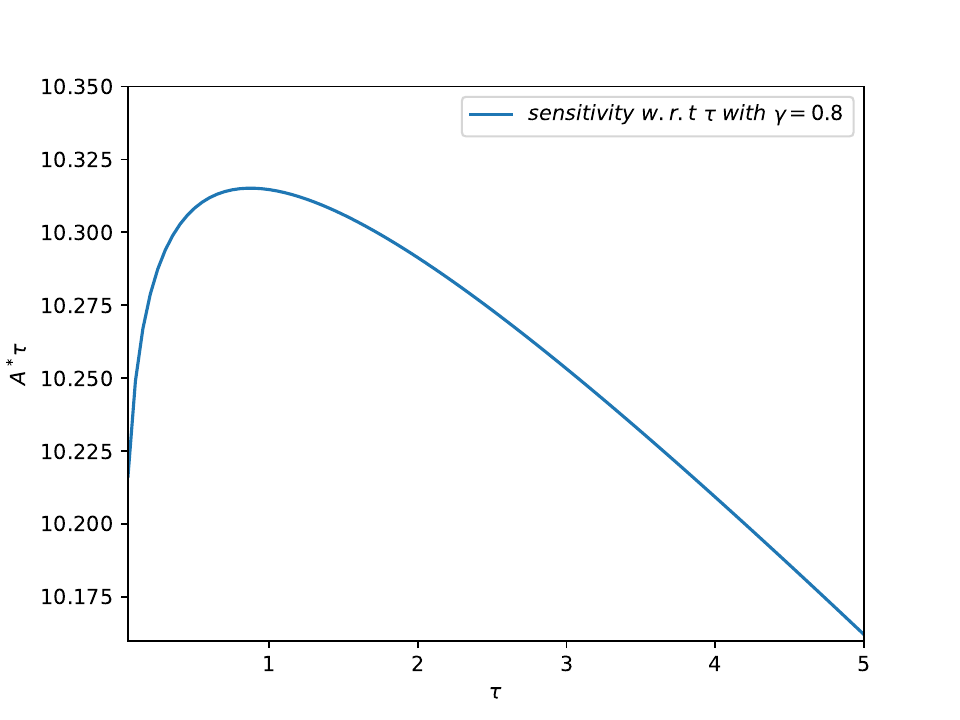}
\quad
\includegraphics[width=6.5cm]{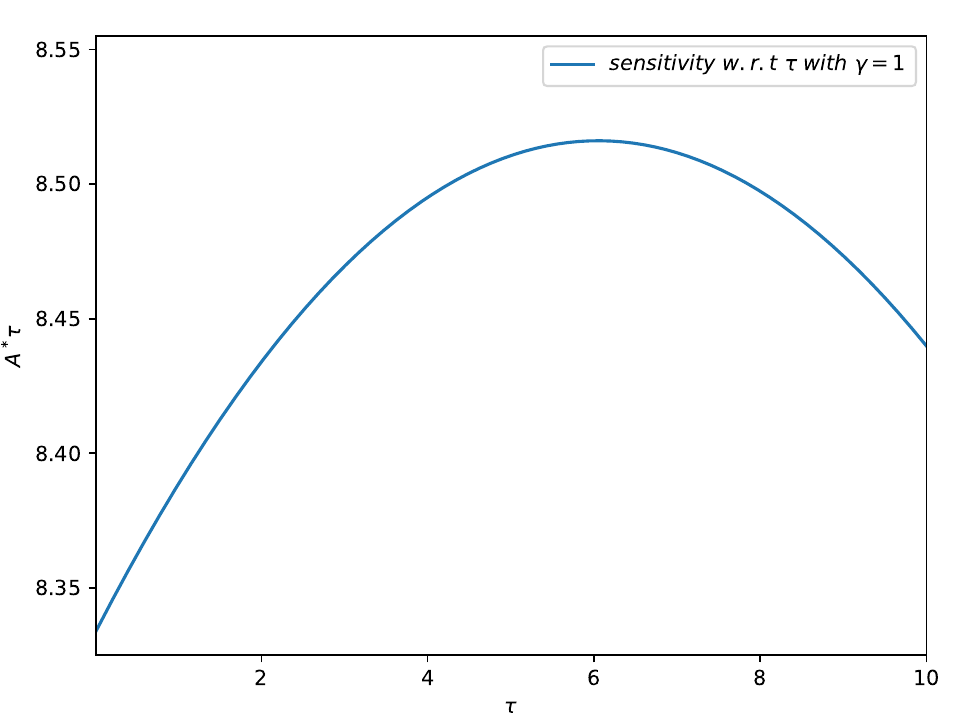}
    \caption{The numerical illustration of $A^*\tau$ with respect to $\tau$ under power utility function when $r=0.12,\,\alpha=0.5$ and $\delta=0.11$}\label{fig-power-tau}
\end{figure}

\ \\

To shed some light on the possible choice of the optimal period length under power utility, we only consider in the next result with the extreme relative performance parameter $\gamma=1$. Note that in this case, the value function reduces to the constant $V(x)=\frac{1}{\alpha}A^*$ independent of the initial wealth $x$ and the operator is simplified to $\Psi(a)=e^{-\delta \tau} \sup_{\pi\in\mathcal{U}_{0}(1)}\mathbb{E}[\frac{1}{\alpha}X_{\tau}^{\alpha}]$. The next result provides the sufficient condition for the optimal choice of the period length $\tau^*$ such that the scaled value function $V(x;\tau)\tau$ can be maximized. 

\begin{prop}\label{taupower}
Assume that the model parameters satisfy $\gamma=1$ and $\frac{\delta}{2}<\zeta(\alpha)<\delta$. Then, there exists an optimal $\tau^*\in\mathbb{R}_+$ such that $$V(x;\tau^*)\tau^*\geq V(x;\tau)\tau,\quad (x,\tau)\in\mathbb{R}_+\times\mathbb{R}_+.$$
\end{prop}

\section{Periodic Evaluation under Logarithmic Utility}\label{section3}
This section aims to study the same portfolio optimization problem in the simple case the logarithmic utility function $U$ that
\begin{eqnarray}\label{log}
U(x):=\log x,\quad x\in\mathbb{R}_+,
\end{eqnarray} 
for which the mathematical arguments are more straightforward.

The goal of the agent with an initial capital $x\in\mathbb{R}_+$ is to maximize the total discounted expected utilities generated by the performance evaluations over all periods, i.e., the value function under the optimal portfolio is defined by
\begin{eqnarray}\label{problem.log}
V(x)
\hspace{-0.3cm}&:=&\hspace{-0.3cm}
\sup_{\pi\in\mathcal{U}_0(x)}\mathbb{E}\left[\sum_{i=1}^{\infty}e^{-\delta T_i}\log\left(\frac{X_{T_i}}{\left(X_{T_{i-1}}\right)^{\gamma}}\right)\right].
\end{eqnarray}

By the Markov property of $X$ and the dynamic programming principle, we can first obtain the relationship that
\begin{eqnarray}\label{ddp.dual.log}
V(x)
\hspace{-0.3cm}&=&\hspace{-0.3cm}
\sup_{\pi\in\mathcal{U}_0(x)}\mathbb{E}\left[e^{-\delta T_1}\log\left(\frac{X_{T_1}}{x^{\gamma}}\right)+e^{-\delta T_1}V(X_{T_1})\right]
\nonumber\\
\hspace{-0.3cm}&=&\hspace{-0.3cm}
\sup_{\pi\in\mathcal{U}_0(x)}\mathbb{E}\left[e^{-\delta \tau}\log\left(\frac{X_{\tau}}{x^{\gamma}}\right)+e^{-\delta \tau}V(X_{\tau})\right],\quad x\in\mathbb{R}_+.
\end{eqnarray}
In view of the logarithmic utility, we heuristically conjecture that the value function satisfies the form of $V(x)=A^*+C^*\log x$ for some constants $A^*,C^*\in\mathbb{R}$ to be determined. Substituting the conjectured expression of $V$ into \eqref{ddp.dual.log} and then subtracting both sides by $C^*\log x$, one can derive that
\begin{eqnarray}\label{A*}
A^*
\hspace{-0.3cm}&=&\hspace{-0.3cm}
\sup_{\pi\in\mathcal{U}_0(x)}\mathbb{E}\left[e^{-\delta T_1}\log\left(\frac{X_{T_1}}{x^{\gamma}}\right)+e^{-\delta T_1}[A^*+C^*\log X_{T_1}]-C^* \log x\right]
\nonumber\\
\hspace{-0.3cm}&=&\hspace{-0.3cm}
\sup_{\pi\in\mathcal{U}_0(x)}\mathbb{E}\left[e^{-\delta \tau}(1+C^*)\log \left(\frac{X_{\tau}}{x}\right)+[C^*(e^{-\delta\tau}-1)+(1-\gamma) e^{-\delta \tau}]\log x\right]+e^{-\delta \tau}A^*.
\end{eqnarray}
By taking $C^*=\frac{1-\gamma}{e^{\delta\tau}-1}$ in the above equation, we can heuristically obtain from \eqref{A*} that
\begin{eqnarray}\label{A*.dual}
A^*=\frac{e^{\delta\tau}-\gamma }{(e^{\delta\tau}-1)^2}\sup_{\pi\in\mathcal{U}_0(1)}\mathbb{E}\left[\log X_{\tau}\right].
\end{eqnarray}
Hence, the heuristic characterization of $V$ boils down to the characterization of the constant $A^*$, which is again to solve the one-period auxiliary terminal wealth utility maximization problem under short-selling constraint with the time horizon $\tau$ that
\begin{eqnarray}\label{log.pri}
V_0(x):=\sup_{\pi\in\mathcal{U}_0(x)}\mathbb{E}\left[\log X_{\tau}\right],\quad x\in\mathbb{R}_+.
\end{eqnarray}

To cope with the auxiliary optimization problem \eqref{log.pri}, we can again employ the martingale duality approach proposed in \cite{XS92(a)}.  However, we note that the duality approach and the main results in \cite{XS92(a)} are not directly applicable in our case because the logarithmic utility does not satisfy their assumptions on the utility function therein (see Definition 2.4 and the following paragraph in \cite{XS92(a)}). 

On the other hand, thanks to the basic property of the logarithmic function, we can write the problem \eqref{problem.log} under periodic evaluation of the relative ratio performance in an equivalent form as 
\begin{align}\label{equlog}
V(x)=\widetilde{V}(x) (1-\gamma e^{-\delta\tau}) - e^{-\delta \tau}\gamma \log x,
\end{align}
where 
\begin{align}\label{equlog-2}
\widetilde{V}(x):=\sup_{\pi\in\mathcal{U}_0(x)}\mathbb{E}\left[\sum_{i=1}^{\infty} e^{-\delta T_i}\log X_{T_i}\right]. 
\end{align}
That is, under the logarithmic utility function, finding the optimal portfolio in the original problem \eqref{problem.log} is identical to solving the portfolio maximization of the sum of expected log utility on terminal wealth with an increasing sequence of time horizons $T_i$, $i=1,2,\ldots$, which itself is a new type of infinite horizon portfolio optimization problem to the literature.

\subsection{Duality Method and Main Results}
We again resort to the martingale duality approach to first study the dual problem associated to \eqref{log.pri} and then establish the duality relationship between optimal solutions to both problems. 

For the logarithmic utility $U(x)=\log x$, we have
\begin{eqnarray}\label{Phi.U}
\Phi_{U}(y):=\sup_{x\geq0}\{U(x)-xy\}=-\log y-1,\quad y\in\mathbb{R}_+.
\end{eqnarray}
Moreover, it is direct to see that $\Phi_{U}(0):=\lim_{y\rightarrow0+}\Phi_{U}(y)=\infty$ and $\Phi_{U}(\infty):=\lim_{y\rightarrow\infty}\Phi_{U}(y)=-\infty$, as well as $\Phi_{U}^{\prime}(y)=-1/y$,  $\Phi_{U}^{\prime\prime}(y)=1/y^2$, $y\in\mathbb{R}_+$.  In summary, the function $\Phi_{U}$ is strictly decreasing, strictly convex and twice differentiable. Recall that the set of dual control processes $\Tilde{\mathcal{U}}$ and the process $Z^{\Tilde{\pi}}$ are defined by \eqref{tilde.U} and \eqref{Z.pi}, respectively.
For any fixed $a,y\in\mathbb{R}_+$ and $\Tilde{\pi}\in\Tilde{\mathcal{U}}_{0,\tau}$, the associated dual functional to the primal problem 
\begin{eqnarray}
\Tilde{J}(y,\Tilde{\pi}):=\mathbb{E}\left[\Phi_{U}(y\frac{Z^{\Tilde{\pi}}_{\tau}}{B_{\tau}})\right]=-\mathbb{E}\left[\log \left(y\frac{Z^{\Tilde{\pi}}_{\tau}}{B_{\tau}}\right)\right]-1,\nonumber
\end{eqnarray}
for any $y\in\mathbb{R}_+$ and $\Tilde{\pi}\in\Tilde{\mathcal{U}}_{0,\tau}$. The dual value function $\Tilde{V}$ is defined by
\begin{eqnarray}\label{dual.pro.log}
\Tilde{V}(y)=\inf_{\Tilde{\pi}\in\Tilde{\mathcal{U}}_{0,\tau}}\Tilde{J}(y,\Tilde{\pi}),\quad y\in\mathbb{R}_+.
\end{eqnarray}

Recall that the constant vector $\Tilde{\pi}^*\in\Tilde{\mathcal{U}}_{0,\tau}$ is the unique minimizer of $\|\xi+\sigma^{-1}\Tilde{\pi}\|$ over $\Tilde{\pi}\in[0,\infty)^n$ and the process $Z^{\Tilde{\pi}^*}_{\tau}$ is given by \eqref{Z.tilde}.
The next result provides the positive answer to the existence of optimal solution to the dual problem \eqref{dual.pro.log}, the duality relationship between the optimal solutions to the auxiliary primal problem \eqref{log.pri} and dual problem \eqref{dual.pro.log}. 
\begin{prop}\label{existence.log} 
The process $\Tilde{\pi}^*$ is  the unique minimizer to  the dual problem \eqref{dual.pro.log}. More specifically, the value function of the dual problem \eqref{dual.pro.log} is given by
\begin{align}
\label{Vtilde}  
\Tilde{V}(y)&=-\mathbb{E}\left[\log\left(y\frac{Z^{\Tilde{\pi}^*}_{\tau}}{B_{\tau}}\right)\right]-1
\nonumber\\
&=-\log y+\frac{1}{2}\|\Tilde{\xi}\|^2\tau+r\tau-1,\quad y\in\mathbb{R}_+.
\end{align}
In addition, the value function of the one-period auxiliary optimization problem \eqref{log.pri} satisfies
\begin{eqnarray}
\label{3.15}
V_0(x)=\sup_{\pi\in\mathcal{U}_{0}(x)}\mathbb{E}[U(X_{\tau})]=\inf_{y\in\mathbb{R}_+}\{\Tilde{V}(y)+xy\}=\Tilde{V}({1}/{x})+1,
\end{eqnarray}
with the optimal wealth process $X^*$ satisfying $X^*_{t}=x\frac{B_{t}}{Z^{\Tilde{\pi}^*}_{t}}$ for $t\in [0,\tau]$. 
\end{prop}

We are next devoted to the verification theorem to identify the solution (i.e., the optimal trading strategy as well as the value function) to the problem \eqref{problem.log} in the next main result of this section. 

\begin{thm}[Verification Theorem]\label{thm3.1} 
The value function of the original problem \eqref{problem.log} is given by
\begin{align}\label{value-log}
V(x)=A^*+C^*\log x, \quad x\in\mathbb{R}_+,
\end{align}
where 
\begin{align}\label{valueparam}
A^*=\frac{e^{\delta\tau}-\gamma}{(e^{\delta\tau}-1)^2}\left(r+\frac{1}{2}\|\Tilde{\xi}\|^2\right)\tau,\quad C^*=\frac{1-\gamma}{e^{\delta\tau}-1},
\end{align}
with $\Tilde{\xi}=\xi+{\sigma}^{-1}{\Tilde{\pi}^*}$ and the constant vector ${\Tilde{\pi}^*}$ being the unique minimizer of $\|\xi+\sigma^{-1}\Tilde{\pi}\|$ over $\Tilde{\pi}\in[0,\infty)^n$. In addition, the optimal wealth process for the infinite horizon original problem is given by
\begin{eqnarray}\label{optimalXlog}
X^*_{t}:=x\frac{B_{t}}{Z_{t}^{\Tilde{\pi}^*}},\quad t\in[0,\infty),  
\end{eqnarray}
where the process $Z_{t}^{\Tilde{\pi}^*}$ is given by \eqref{Z.tilde}, and the optimal portfolio process admits the feedback form that
\begin{eqnarray}\label{optimalpilog}
\pi^*_t=\pi^*(X_t^*),
\end{eqnarray}
with the linear feedback function
\begin{eqnarray}\label{pilogopt}
\pi^*(x):=x(\sigma^{\mathsf{T}})^{-1}\left[\xi+{\sigma}^{-1}{\Tilde{\pi}^*}\right],\quad t\in [0,\infty).
\end{eqnarray}
\end{thm}

\subsection{Quantitative Properties and Numerical Illustrations}
Unlike the power utility, the optimal portfolio process and the optimal wealth in the logarithmic utility have explicit expressions. Based on the explicit characterizations in Theorem \ref{thm3.1}, we attempt to numerically illustrate in this subsection some quantitative properties of the value function and the optimal portfolio process and discuss some financial insights stemming from the periodic reward structure and the short-selling constraint. 

Again, we choose the mean rates of return $\mu=(\mu_1,\mu_2)^{\mathsf{T}}$ and the $2\times2$ volatility matrix $\sigma$ by
\begin{eqnarray}
    \mu=\left(
    \begin{array}{cc}
         0.1 \\
         0.15
    \end{array}\right),\quad
    \sigma=\left(
    \begin{array}{cc}
        0.2 & 0 \\
        0 & 0.25
    \end{array}
    \right),\nonumber
\end{eqnarray}
and other market parameters are listed in Table \ref{tab1} (unless otherwise specified):
\begin{table}[H]
\centering
\caption{The values of the parameters.}
\label{tab1}
 \begin{tabular}{lccccc}
\hline\hline
\noalign{\smallskip}
Parameter & $r$ & $\tau$ & $\gamma$  & $\delta$ &$x$ \\\noalign{\smallskip}
Value & 0.12& 1 & 0.8 &  0.3 &0.5 \\\noalign{\smallskip}
\hline
\end{tabular}   
\end{table}
\vspace{0.2in}

\begin{figure}[H]
    \centering
\includegraphics[width=6.5cm]{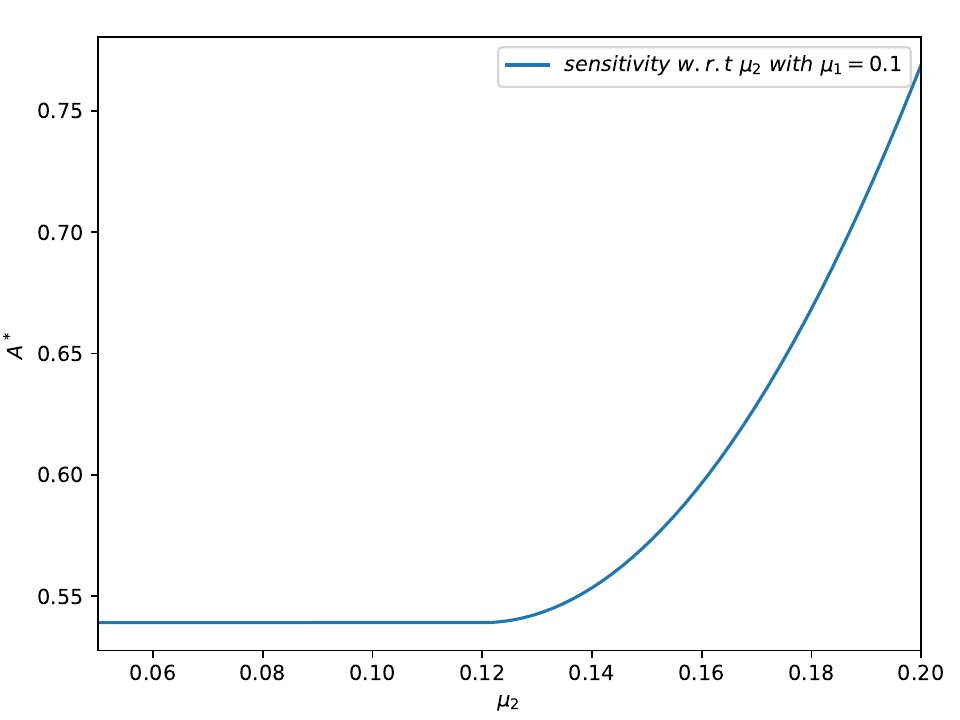}
    \quad
\includegraphics[width=6.5cm]{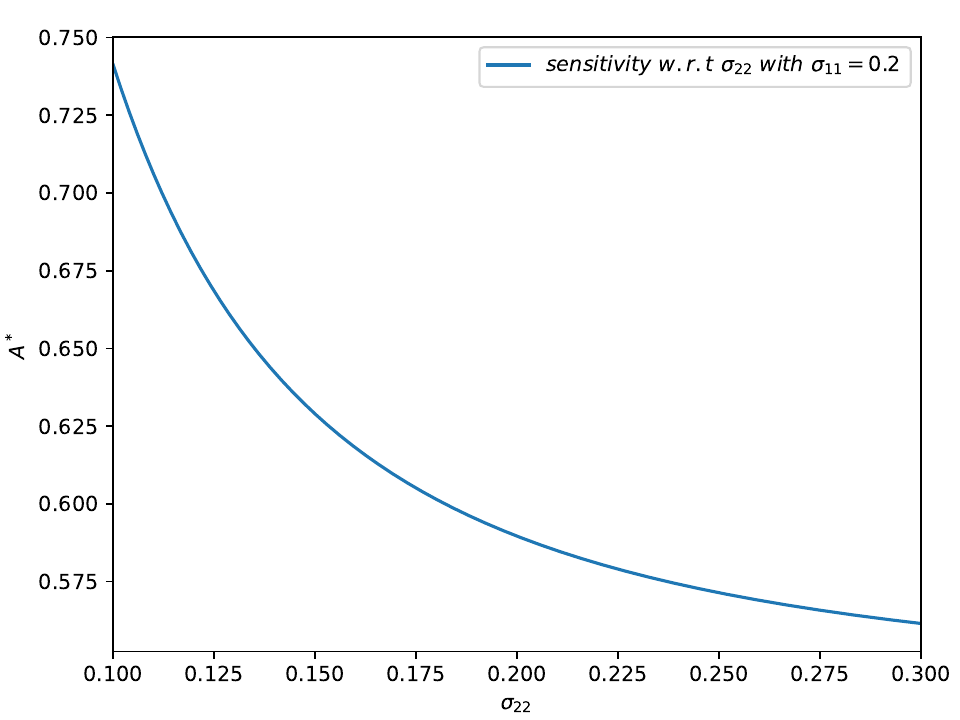}\\
\includegraphics[width=6.5cm]{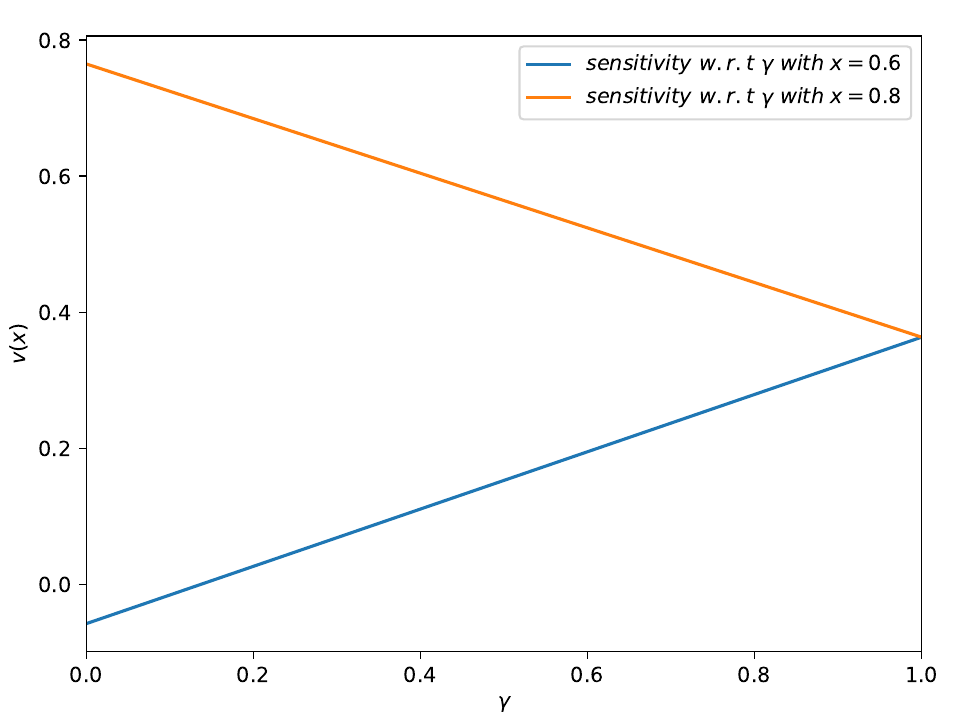}
    \caption{The sensitivity of $A^*$ with respect to the drift $\mu$ (left-top panel) and the volatility $\sigma$ (right-top panel); The sensitivity of $V(x)$ with respect to $\gamma$ (bottom panel)}\label{fig-log-mu}
\end{figure}

Recall that $V(x)=A^*+C^*\log x$, where only $A^*$ depends on the mean return $\mu$ and the volatility $\sigma$. To illustrate the sensitivity of $V(x)$ on these two parameters, it is sufficient to plot the sensitivity of $A^*$ with respect to $\mu$ and $\sigma$ in the top panels of Figure \ref{fig-log-mu}. We can observe from the graphs that the agent does not invest in the corresponding stock when its expected return $\mu$ is lower than $r$ due to the short-selling constraint; and once the agent allocates positive wealth into the stock, the value $A^*$ (and hence the value function $V(x)$) is increasing in the expected return and decreasing in the volatility of the stock, which are consistent with results in classical models with a finite time horizon. In addition, we also plot the sensitivity of the value function $V(x)$ with respect to the relative performance parameter $\gamma$ in the bottom panel of Figure \ref{fig-log-mu}. It is not surprising that the value function $V(x)$ is linear in the parameter $\gamma$ from the problem formulation. When the agent holds a low initial wealth, the value function is decreasing in $\gamma$; When the agent starts with a sufficient initial wealth, the larger relative performance parameter increases the value function. These trends can also be directly observed from the equivalent problem formulation in \eqref{equlog} as the problem \eqref{equlog-2} is independent of $\gamma$. However, from \eqref{optimalXlog} and \eqref{optimalpilog} or \eqref{equlog-2}, the constrained optimal portfolio $\pi_t^*$ and the optimal wealth process $X_t^*$ are actually independent of the parameter $\gamma$. Therefore, under the logarithmic utility, the relative performance parameter $\gamma$ has adverse effect in the value function but it in fact has no impact on the choice of the optimal portfolio. In practice, the logarithmic utility agent can freely choose the value of $\gamma\in(0,1]$ and will arrive at the same decision making.

Next, we are interested in examing the impact cause by the short-selling prohibition under the periodic evaluation. To this end, let $V_n$ denote the value function of the investment problem under no portfolio constraint. By mimicking the method used for deriving the value function of the optimization problem \eqref{problem.log} in the previous subsection, if short-selling is allowed, one can obtain that the unconstrained value function satisfies
\begin{eqnarray}
V_n(x)=A+C^* \log x,\quad x\in\mathbb{R}_+,\nonumber
\end{eqnarray}
with
$$A:=\frac{e^{\delta\tau}-\gamma}{(e^{\delta\tau}-1)^2}\sup_{\pi\in\mathcal{A}_0(1)}\mathbb{E}[\log X_{\tau}],\quad C^*=\frac{1-\gamma}{e^{\delta\tau}-1},$$
where $\sup_{\pi\in\mathcal{A}_0(1)}\mathbb{E}[\log X_{\tau}]$ is the value function of the one-period utility maximization problem without the portfolio constraint, and the set of admissible controls is defined by
\begin{eqnarray}\label{set.At}
\mathcal{A}_0(x)
\hspace{-0.3cm}&:=&\hspace{-0.3cm}
\bigg\{\pi:\pi=(\pi^1_t,...,\pi^n_t)_{t\geq0}\text{ is a predictable and locally square-integral process such that}
\nonumber\\
\hspace{-0.3cm}&&\hspace{-0.3cm}
\text{ }X_t=x+\sum_{i=1}^n\int_0^t{\pi^i_u}\frac{dS^i_u}{S^i_u}+\int_0^t(X_u-\mathbf{1}^{\mathsf{T}}\pi_u)\frac{dB_u}{B_u}>0 \text{ for } t\geq 0,
\nonumber\\
\hspace{-0.3cm}&&\hspace{-0.3cm}
\color{black}\text{and } 
\sum_{i=1}^{\infty}e^{-\delta T_i}\mathbb{E}\left[\left(U\left(\frac{X_{T_i}}{\left(X_{T_{i-1}}\right)^{\gamma}}\right)\right)_{-}\right]<\infty
\bigg\},\quad x>0.\nonumber
\end{eqnarray}
Consequently, one can follow the arguments of Section 6 in \cite{KLSX91} to verify that
$$\sup_{\pi\in\mathcal{A}_0(1)}\mathbb{E}[\log X_{\tau}]=\mathbb{E}\left[\log\frac{B_{\tau}}{Z_{\tau}}\right]=r\tau+\frac{1}{2}\|\xi\|^2\tau,$$
with process $Z$ defined in \eqref{z.process}.
Hence, similar to Theorem \ref{thm3.1}, the next result follows.

\begin{corollary}\label{corshortlog}
When the short-selling is allowed, the value function under the period evaluation satisfies
$$V_n(x)=\frac{e^{\delta\tau}-\gamma}{(e^{\delta\tau}-1)^2}\left(r+\frac{1}{2}\|\xi\|^2\right)\tau+\frac{1-\gamma}{e^{\delta\tau}-1}\log x\quad x\in\mathbb{R}_+.$$
In addition, the optimal wealth process $X_t^{*,n}$ and the feedback optimal portfolio $\pi_t^{*,n}$ satisfy 
\begin{align}\label{canshortpi}
&X_t^{*,n}=x\frac{B_t}{Z_t},\quad \text{and}\quad \pi_t^{*,n}=\pi^{*,n}(X_t^{*,n}),\notag\\
&\text{where}\ \ \pi^{*,n}(x)=x(\sigma\sigma^{\mathsf{T}})^{-1}(\mu-r\mathbf{1}).
\end{align}
\end{corollary}

\begin{rem}
First, we note that the feedback functions in \eqref{pilogopt} indicate that it is optimal to allocate either zero (due to short-selling constraint) or a fixed positive fraction of wealth into the stocks. Similarly, when short-selling is allowed, \eqref{canshortpi} also suggest a fixed proportion of wealth to invest in stocks. In fact, in both cases with and without short-selling constraint, we can also study the classical finite horizon utility maximization problem on terminal wealth 
\begin{align*}
\sup_{\pi^1\in\mathcal{A}^1(x)}\mathbb{E}[\log(X_T)]
\end{align*}
and the infinite horizon utility maximization problem on the intermediate consumption that $$\sup_{(\pi^2,c)\in\mathcal{A}^2(x)}\mathbb{E}\left[\int_0^\infty \log(c_t)dt\right],$$
where $\mathcal{A}^1(x)$ and $\mathcal{A}^2(x)$ stand for the sets of admissible controls in the conventional setting.
Although these two classical problems differ from our optimization problems under periodic evaluation, under the special logarithmic utility, it is interesting to see that the optimal portfolio processes for these two classical problems actually coincide with the optimal solution in our new problem with periodic evaluation (in either case with short-selling constraint or without short-selling constraint), or equivalently the constant wealth fraction optimal portfolio in classical problems
coincides with the optimal solution in the new infinite horizon portfolio optimization problem \eqref{equlog-2}. Therefore, under the logarithmic utility, the periodic evaluation criterion leads to a distinct value function comparing with the classical Merton's problems, but it produces the same optimal portfolio strategy as if in the classical problem without periodic evaluation. 

However, we also stress that the coincidence of optimal solutions in different problems is only consequent on the property of logarithmic utility. It was shown in Section \ref{section4} that the optimal portfolios in the classical problems and in our new formulation with periodic evaluation differ significantly when the power utility is adopted.    

\end{rem}

From the definitions of $\xi$ and $\Tilde{\xi}$ in \eqref{xi} and \eqref{tilde.xi}, it readily follows that
$$\|\Tilde{\xi}\|=\inf_{\Tilde{\pi}\in[0,\infty)^n}\|\xi+\sigma^{-1}\Tilde{\pi}\|\leq \|\xi+\sigma^{-1}\Tilde{\pi}\||_{\Tilde{\pi}\equiv0}=\|\xi\|,$$
and the next result is a direct consequence of Theorem \ref{thm3.1} and Corollary \ref{corshortlog}. 

\begin{corollary}\label{costshortsell}
The difference between two value functions is a positive constant that
\begin{align}\label{comp-diff}
    V_n(x)-V(x)=A-A^*=\frac{e^{\delta\tau}-\gamma}{2(e^{\delta\tau}-1)^2}\left(\|\xi\|^2-\|\Tilde{\xi}\|^2\right)\tau\geq0,\quad x\in\mathbb{R}_+.
\end{align}
\end{corollary}
The positive difference in \eqref{comp-diff} can be interpreted as the cost of short-selling constraint under our proposed periodic evaluation when the logarithmic utility is adopted. It is interesting to see that this cost under the logarithmic utility is independent of the initial wealth $x$, and it is linear and decreasing in the relative performance parameter $\gamma$. So, if the agent accounts more on the relative performance with the larger $\gamma$, the impact of the short-selling constraint is more weakened.  

%
%

We stress that for logarithmic utility agent, the optimal wealth process $X_t^*$ in \eqref{optimalXlog} and the optimal constrained portfolio process in \eqref{optimalpilog} are actually independent of the period length $\tau$. As a result, in terms of the optimal decision making, it does not matter which period length $\tau$ the agent adopts and the period length in practice might be determined in the appraisal system. On the other hand, if the agent has the authority in choosing the period length for periodic evaluation, one natural question arises: based on the explicit value function $V(x)$ in \eqref{value-log} and \eqref{valueparam}, is it possible to discuss the optimal periodicity of $\tau$ to maximize some targeted profit? 

Thanks to the nice property of logarithmic function, this dependence on $\tau$ become transparent under some circumstances, which makes the value function as a natural choice to study  the optimal periodicity. The next result states that under certain sufficient conditions on model parameters and the initial wealth $x$, there exists an optimal $\tau^*$ attaining the maximal value function $\sup_{\tau} V(x;\tau)$.

\begin{prop}\label{optimaltau-1}
If $\gamma \in(0,1)$ and $\frac{r+\frac{1}{2}\|\Tilde{\xi}\|^2}{\delta}+\log x<0$, there exists an optimal $\tau^*$ such that
$$V(x;\tau^*)\geq V(x;\tau)\vee0,\quad\textrm{for any}\ \ (x,\tau)\in\mathbb{R}_+\times\mathbb{R}_+.$$
\end{prop}

In particular, for the given market model, if the initial wealth $x$ is very low, the sufficient condition in Proposition \ref{optimaltau-1} can be easily validated. On the other hand, when the condition in Proposition \ref{optimaltau-1} is violated, we may not be able to strategically choose $\tau$ by maximizing the value function directly. However, from the perspective to scale the frequency in the periodic evaluation, it is also reasonable to consider the target as the scaled value function $V(x;\tau)\tau$. The next result provides positive answers to the choice of the optimal $\tau^*$ when the condition in Proposition \ref{optimaltau-1} can not be fulfilled, that is, we can choose $\tau$ to maximize the scaled value function $V(x;\tau)\tau$ instead of $V(x,\tau)$; See the numerical illustrations in Figure \ref{fig-tau}.

\begin{figure}[H]
    \centering
\includegraphics[width=6.5cm]{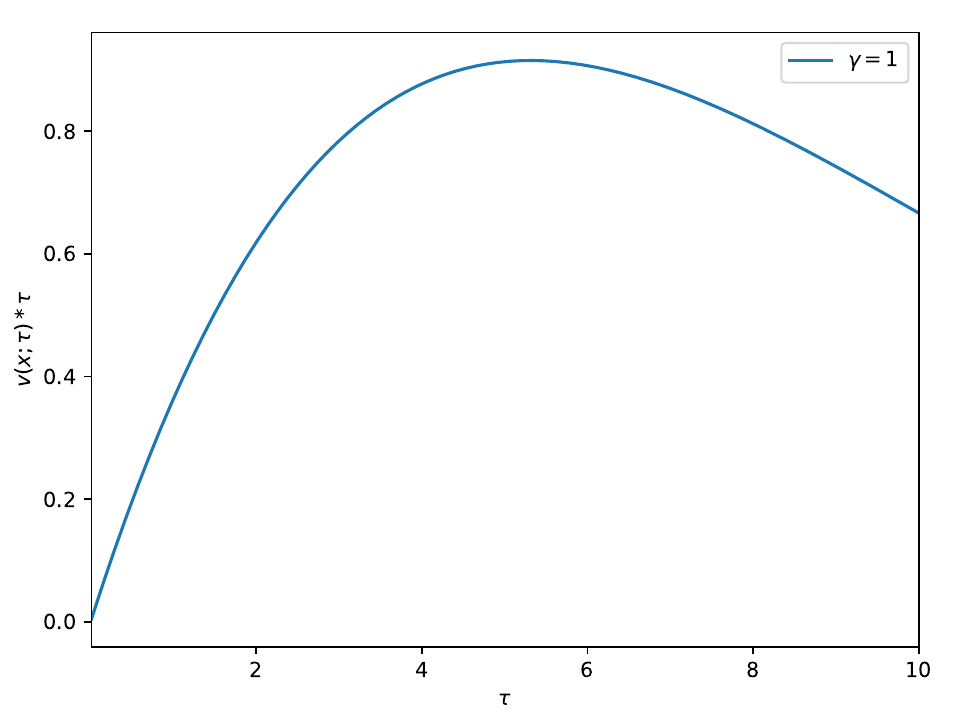}
    \quad
\includegraphics[width=6.5cm]{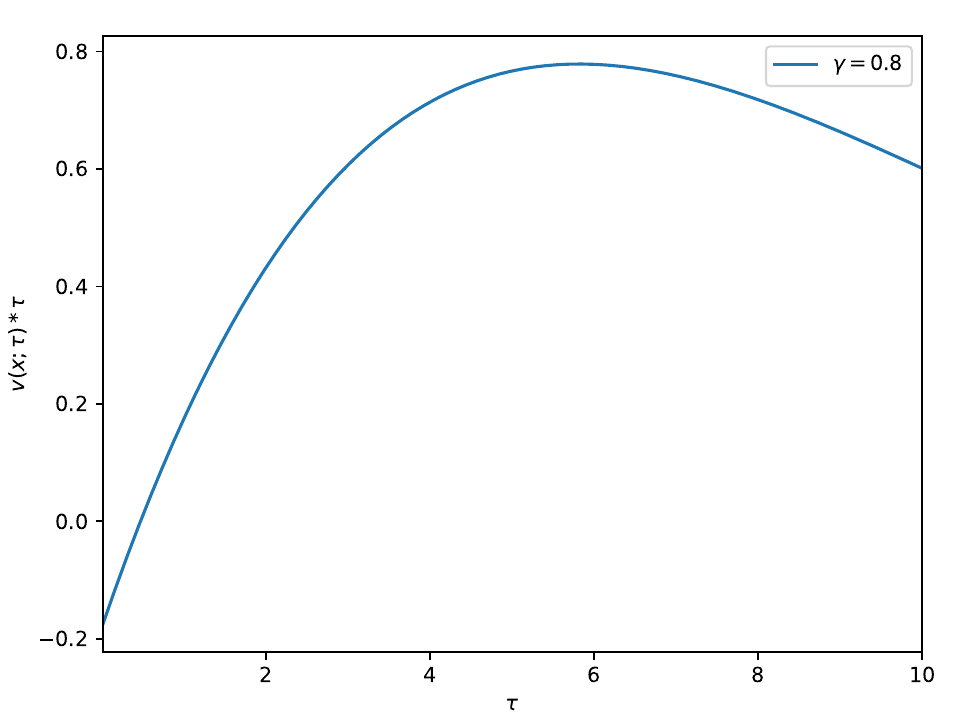}
    \caption{The plots of $V(x;\tau)\tau$ with respect to $\tau$ when $\gamma=1$ and $\gamma=0.8$ under the logarithmic utility function}\label{fig-tau}
\end{figure}

\begin{prop}\label{optimaltau-2}
\begin{itemize}
    \item[(1)] If $\gamma=1$, there exists an optimal $\tau^*$ such that $$V(x;\tau^*)\tau^*\geq \left(V(x;\tau)\tau\right)\vee 0,\quad\textrm{for any}\ \ (x,\tau)\in\mathbb{R}_+\times\mathbb{R}_+.$$
    \item[(2)] If $\gamma\in(0,1)$ and $\left(r+\frac{1}{2}\|\Tilde{\xi}\|^2\right)\frac{\gamma}{\delta}-\frac{1-\gamma}{2}\log x>0$,  there exists an optimal $\tau^*$ such that $$V(x;\tau^*)\tau^*\geq V(x;\tau)\tau,\quad\textrm{for any}\ \ (x,\tau)\in\mathbb{R}_+\times\mathbb{R}_+.$$
\end{itemize}  
\end{prop}

As a lesson from Propositions \ref{optimaltau-1} and \ref{optimaltau-2}, when the agent starts with a low financial situation, he is suggested to choose the optimal period length $\tau^*$ strategically (whenever it is possible) for periodic evaluation to maximize the value function $V(x;\tau)$; when the agent possesses abundant initial wealth, he is instead recommended to seek the optimal period length $\tau^*$ (whenever it is possible) such that the scaled value function $V(x;\tau)\tau$ can attain its maximum. Admittedly, to better discuss the optimal choice of periodicity $\tau^*$ in practice, we also need to incorporate the a financially sound cost function $F(\cdot)$ of the evaluation period $\tau$ into consideration, which can be formulated as an interesting optimization problem in the context of periodic evaluation.

\section{Proofs}\label{sec-proof}

\subsection{Proofs in Section \ref{section4}}

The following result gives the upper and lower bound for the value function $V$ in \eqref{problem}, which will be used in the later analysis. 

\begin{lem}\label{boundlem}
It holds that
\begin{eqnarray}
\label{2.12}
\frac{ e^{(r\alpha-\delta)\tau}}{\alpha(1-e^{-(\delta-r\alpha(1-\gamma))\tau})}
x^{\alpha(1-\gamma)}\leq V(x)\leq
\frac{e^{(\zeta(\alpha)-\delta)\tau}}{\alpha(1-e^{(\zeta(\alpha(1-\gamma))-\delta)\tau})}
x^{\alpha(1-\gamma)},\quad x\in\mathbb{R}_+.
\end{eqnarray}
\end{lem}

\begin{proof}
Since the proofs for the cases $\alpha\in(0,1)$ and $\alpha\in(-\infty,0)$ are similar, we will only present the proof for the situation where $\alpha\in(0,1)$.
The lower bound of $V(x)$ can be obtained by the fact that $\pi=(\pi_t)_{t\geq0}\equiv 0$ is an admissible portfolio process in $\mathcal{U}_0(x)$.
We next show the upper bound. For any $\pi\in\mathcal{U}_0(x)$ and the associated wealth process $X$, we have
\begin{eqnarray}\label{p2.E<E}
\hspace{-0.4cm}
\mathbb{E}\left[\sum_{i=1}^{\infty}e^{-\delta T_i}U\bigg(\frac{X_{T_i}}{(X_{T_{i-1}})^{\gamma}}\bigg)\right]
\hspace{-0.3cm}&=&\hspace{-0.3cm}
\sum_{i=1}^{\infty}e^{-\delta T_{i}}\mathbb{E}\left[\frac{1}{\alpha}\bigg(\frac{X_{T_i}}{X_{T_{i-1}}}\bigg)^{\alpha}X_{T_{i-1}}^{\alpha(1-\gamma)}\right]
\nonumber\\
\hspace{-0.3cm}&=&\hspace{-0.3cm}
\sum_{i=1}^{\infty}e^{-\delta T_{i}}\mathbb{E}\left[\mathbb{E}\left[\frac{1}{\alpha}\bigg(\frac{X_{T_i}}{X_{T_{i-1}}}\bigg)^{\alpha}\bigg|\mathcal{F}_{T_{i-1}}\right]X_{T_{i-1}}^{\alpha(1-\gamma)}\right]
\nonumber\\
\hspace{-0.3cm}&\leq&\hspace{-0.3cm}
\sum_{i=1}^{\infty}e^{-\delta T_{i}}\mathbb{E}\left[\left(\sup_{\pi\in \mathcal{U}_{T_{i-1}}(1)}\mathbb{E}\left[\left.\frac{1}{\alpha}X_{T_i}^{\alpha}\right|\mathcal{F}_{T_{i-1}}\right]\right)X_{T_{i-1}}^{\alpha(1-\gamma)}\right].
\end{eqnarray}
Note that $\sup_{\pi\in\mathcal{U}_{T_{i-1}}(1)}\mathbb{E}\left[\frac{1}{\alpha}\left.X_{T_i}^{\alpha}\right|\mathcal{F}_{T_{i-1}}\right]$ is the value function of the utility maximization problem on terminal wealth under short-selling prohibition with power utility function, the maturity $\tau$, and the initial wealth of 1. Consequently, we can adopt the results of Section 4 in \cite{XS92(b)} to derive that
\begin{eqnarray}\label{p2.supE<e}
\sup_{\pi\in\mathcal{U}_{T_{i-1}}(1)}\mathbb{E}\left[\left.\frac{1}{\alpha}X_{T_i}^{\alpha}\right|\mathcal{F}_{T_{i-1}}\right]= \frac{1}{\alpha}e^{\zeta(\alpha)\tau}.
\end{eqnarray}
Similarly, we also have that
\begin{eqnarray}\label{p2.supE<e2}
\sup_{\pi\in\mathcal{U}_0(x)}\mathbb{E}\left[\frac{1}{\alpha}X_{T_{i-1}}^{\alpha(1-\gamma)}\right]=\frac{1}{\alpha} x^{\alpha(1-\gamma)}e^{\zeta(\alpha(1-\gamma))(i-1)\tau}.
\end{eqnarray}
Combining \eqref{p2.E<E}-\eqref{p2.supE<e2}, we can conclude that
\begin{eqnarray}
\mathbb{E}\left[\sum_{i=1}^{\infty}e^{-\delta T_i}U\bigg(\frac{X_{T_i}}{(X_{T_{i-1}})^{\gamma}}\bigg)\right]
\hspace{-0.3cm}&\leq&\hspace{-0.3cm}
\sum_{i=1}^{\infty}e^{-\delta T_{i}}\mathbb{E}\left[\left(\sup_{\pi\in \mathcal{U}_{T_{i-1}}(1)}\mathbb{E}\left[\left.\frac{1}{\alpha}X_{T_i}^{\alpha}\right|\mathcal{F}_{T_{i-1}}\right]\right)X_{T_{i-1}}^{\alpha(1-\gamma)}\right]
\nonumber\\
\hspace{-0.3cm}&\leq&\hspace{-0.3cm}
\sum_{i=1}^{\infty}e^{-\delta T_{i}}e^{\zeta(\alpha)\tau}\sup_{\pi\in\mathcal{U}_0(x)}\mathbb{E}\left[\frac{1}{\alpha}X_{T_{i-1}}^{\alpha(1-\gamma)}\right]
\nonumber\\
\hspace{-0.3cm}&=&\hspace{-0.3cm}
\frac{e^{(\zeta(\alpha)-\delta)\tau}}{\alpha(1-e^{(\zeta(\alpha(1-\gamma))-\delta)\tau})}x^{\alpha(1-\gamma)},\nonumber
\end{eqnarray}
which is the desired upper bound in \eqref{2.12}.
\end{proof}

\ \\
\begin{proof}[Proof of Lemma \ref{lem2.1}]
For the first claim, let us first differentiate both sides of \eqref{2.8} to get that
\begin{eqnarray}\label{h_a'}
h_a^{\prime}(x)= x^{\alpha-1}+a(1-\gamma)x^{\alpha(1-\gamma)-1},\quad x\in\mathbb{R}_+,
\end{eqnarray}
and
\begin{eqnarray}
h_a^{\prime\prime}(x)=(\alpha-1) x^{\alpha-2}+a(1-\gamma)(\alpha-1-\alpha\gamma)x^{\alpha(1-\gamma)-2},\quad x\in\mathbb{R}_+.\nonumber
\end{eqnarray}
Recalling that $a\in\mathbb{R}_+$ and $\gamma\in(0,1]$, we have $h_a^{\prime}(x)>0$ and $h_a^{\prime\prime}(x)<0$ for all $x\in\mathbb{R}_+$, implying that the function $h_a(x)$ is strictly concave and strictly increasing on $(0,\infty)$. Furthermore, by \eqref{h_a'}, one can easily get $h_a^{\prime}(0+)=\infty$ and $h_a^{\prime}(\infty)=0$. 

To prove the second claim, for any constant $\varrho\in(1,\infty)$, one can take $\vartheta=\varrho^{\alpha-1}\in(0,1)$. Then 
\begin{align}
    \vartheta h_a^{\prime}(x)&=(\varrho x)^{\alpha-1}+a\alpha(1-\gamma)\varrho^{\alpha-1}x^{\alpha(1-\gamma)-1}
    \nonumber\\
    &=h_a^{\prime}(\varrho x)+a\alpha(1-\gamma)(\varrho x)^{\alpha(1-\gamma)-1}(\varrho^{\alpha\gamma}-1)
    \nonumber\\
    &\geq h_a^{\prime}(\varrho x),\quad x\in\mathbb{R}_+,\nonumber
\end{align}
where the last inequality follows from $\gamma\in(0,1]$ and $\varrho^{\alpha\gamma}\in(1,\infty)$.
Finally, when $\alpha\in(0,1)$, it follows from $h_a(0)=0$ and $h_a^{\prime}(x)>0$ for $x\in\mathbb{R}_+$ that $h_a(x)>0$ for $x\in\mathbb{R}_+$;
on the other hand, taking $\kappa_1=2\max\{1/\alpha,a\}\in(0,\infty)$ and $\rho_1=\alpha\in(0,1)$ yields \eqref{h_a<kappa}. The proof is complete.
\end{proof}

\ \\
\begin{proof}[Proof of Proposition \ref{existence}]
Note that the function $[0,\infty)^n\ni\Tilde{\pi}\mapsto\|\xi+\sigma^{-1}\Tilde{\pi}\|\in[0,\infty)$ is continuous and convex with respect to $\Tilde{\pi}$, and satisfies $\lim_{\|\Tilde{\pi}\|\rightarrow\infty}\|\xi+\sigma^{-1}\Tilde{\pi}\|=\infty$. Hence, it has a unique minimizer, say, $\Tilde{\pi}^*\in[0,\infty)^n$. We next apply an adapted version of Theorem 2.1 of \cite{XS92(b)} to explicitly characterize an optimal solution $\Tilde{\pi}_y\in\Tilde{\mathcal{U}}_{0,\tau}$ to the dual problem \eqref{dual.pro} with initial condition $y$.
We first consider the case when $\alpha\in(0,1)$.
By \eqref{h_a<kappa} and \eqref{phi(y)=h}, one can easily verify that
$$0\leq \Phi_{h_a}(y)\leq \kappa_2(1+y^{-\frac{\rho_{1}}{1-\rho_{1}}}),$$
which, combined with \eqref{dual.pro} and the fact that $\xi$ is uniformly bounded constant vector, gives
\begin{align}
\label{dualfun.finite}
    \Tilde{V}_a(y)&\leq \Tilde{J}_a\big(y,\Tilde{\pi}\equiv\vec{0}\big)=\mathbb{E}\left[\Phi_{h_a}(y\frac{Z^{\vec{0}}_{\tau}}{B_{\tau}})\right]
    \nonumber\\
    &\leq \kappa_{2}\left(1+\left(\frac{y}{B_{\tau}}\right)^{-\frac{\rho_{1}}{1-\rho_{1}}}\mathbb{E}\left[
\exp\left(\frac{\rho_{1}}{1-\rho_{1}}\xi^{\mathsf{T}}W_t+\frac{\rho_{1}}{2-2\rho_{1}}\|\xi\|^2t\right) 
    \right]\right)
    \nonumber\\
    &<\infty, \quad a\in\mathbb{R}_{+}, \,\,y\in\mathbb{R}_{+}.
\end{align}
With the help of \eqref{dualfun.finite}, one can rewrite the control problem \eqref{dual.pro} as \begin{eqnarray}\label{dual.pro.0}
\Tilde{V}_a(y)=\inf_{\Tilde{\pi}\in\Tilde{\mathcal{V}}_{0,\tau}}\Tilde{J}_a(y,\Tilde{\pi}),\quad y\in\mathbb{R}_+,
\end{eqnarray}
where $\Tilde{\mathcal{V}}_{0,\tau}:=\{\pi: \pi\in\Tilde{\mathcal{U}}_{0,\tau}\text{ and } \mathbb{E}\left[\Phi_{h_a}(y\frac{Z^{\Tilde{\pi}}_{\tau}}{B_{\tau}})\right]<\infty\}$ with $\Tilde{\mathcal{U}}_{0,\tau}$ defined in \eqref{tilde.U}. Mimicking similar arguments as those used in \eqref{dualfun.finite}, one can find that $\Tilde{\pi}^*\in \Tilde{\mathcal{V}}_{0,\tau}$.
In addition, following from \eqref{phi(infty)} and \eqref{Phi'}, the function $\mathbb{R}_{+}\ni y\mapsto \Phi_{h_a}(y)$ is non-increasing, lower-bounded and convex. Then, using Lemma 5.1, Theorems 2.1 and 5.2 of \cite{XS92(b)}, one can verify that 
$$\Tilde{J}_a(y,\Tilde{\pi})\geq \Tilde{J}_a(y,\Tilde{\pi}^{*}), \quad \Tilde{\pi}\in\Tilde{\mathcal{V}}_{0,\tau},$$
which together with \eqref{dual.pro.0} yields the desired result.
Note that $\Phi_{h_a}$ is non-negative, continuous, non-increasing and convex with respect to $y\in\mathbb{R}_+$ when $\alpha\in(0,1)$, so is $\Tilde{V}_a$ as a result.
The fact that $\Tilde{V}_a$ is finite has been verified by \eqref{dualfun.finite}. Moreover, it follows from the monotone convergence theorem
as well as \eqref{phi(0)}-\eqref{phi(infty)}
that $$\lim_{y\rightarrow0+}\mathbb{E}\left[\Phi_{h_a}(y\frac{Z^{\Tilde{\pi}^*}_{\tau}}{B_{\tau}})\right]=\infty,\text{ and }\lim_{y\rightarrow\infty}\mathbb{E}\left[\Phi_{h_a}(y\frac{Z^{\Tilde{\pi}^*}_{\tau}}{B_{\tau}})\right]=0.$$ 

We next consider the case when $\alpha\in(-\infty,0)$. Note that one can follow similar arguments of Lemma 5.1 in \cite{XS92(b)} to prove the submartingale property of $(\Phi_{h_a}(ye^{-r\tau}Z^{\tilde{\pi}^*}_{t}))_{t\geq0}$ via defining the bounded, nonincreasing, convex function $\Phi_{\epsilon}$ by
\begin{eqnarray}
\Phi_{\epsilon}(x)=\left\{
\begin{array}{ll}
\Phi_{h_a}(ye^{-r\tau}x),\quad &x\in(0,\epsilon),\\
\Phi_{h_a}(ye^{-r\tau}\epsilon),\quad &x\in[\epsilon,\infty).
\end{array}
   \right.\nonumber
\end{eqnarray}
Then, we can apply the same arguments of Theorem 2.1 and 5.2 in \cite{XS92(b)} to verify that the unique minimizer $\Tilde{\pi}^*$ of the function $[0,\infty)^n\ni\Tilde{\pi}\mapsto\|\xi+\sigma^{-1}\Tilde{\pi}\|$ is an optimal solution to the dual problem \eqref{dual.pro} with initial condition $y$. Note that $\Phi_{h_a}$ is non-positive, continuous, non-increasing and convex with respect to $y\in\mathbb{R}_+$ when $\alpha\in(-\infty,0)$, so is $\Tilde{V}_a$ as a result. This, combined with the fact that $Z^{\Tilde{\pi}^*}$ is a supermartingale, implies
\begin{eqnarray}
\Tilde{V}_a(y)=\mathbb{E}\left[\Phi_{h_a}(y\frac{Z^{\Tilde{\pi}^*}_{\tau}}{B_{\tau}})\right]\geq \Phi_{h_a}\left(y\frac{\mathbb{E}[Z^{\Tilde{\pi}^*}_{\tau}]}{B_{\tau}}\right)\geq\Phi_{h_a}\left(y\frac{\mathbb{E}[Z^{\Tilde{\pi}^*}_{0}]}{B_{\tau}}\right)>-\infty,\nonumber
\end{eqnarray}
which implies that $\Tilde{V}_a$ is finite. Moreover, it follows from the monotone convergence theorem as well as 
\eqref{phi(0)}-\eqref{phi(infty)} that 
$$\lim_{y\rightarrow0+}\mathbb{E}\left[\Phi_{h_a}(y\frac{Z^{\Tilde{\pi}^*}_{\tau}}{B_{\tau}})\right]=0,\text{ and }\lim_{y\rightarrow\infty}\mathbb{E}\left[\Phi_{h_a}(y\frac{Z^{\Tilde{\pi}^*}_{\tau}}{B_{\tau}})\right]=-\infty.$$ 
The proof is complete.
\end{proof}

\ \\
The following result shows that the function $\ell_{a,y}(\cdot)$ is differentiable at 1 for any fixed $a,y\in\mathbb{R}_+$. 
Note that, combining the property \eqref{h_a'>h_a'} of the function $h_a$ and the fact that $\Tilde{\pi}^*$ is the optimal solution for the dual problem \eqref{dual.pro}, we can apply a similar argument as that in the proof of Lemma 4.6 in \cite{XS92(a)} to obtain the next result, and hence, its proof is omitted.

\begin{lem}\label{ell'(1).equ}
For any fixed $a,y\in\mathbb{R}_+$, the function $\ell_{a,y}$ defined by \eqref{ell} is differentiable at 1 and
\begin{eqnarray}\label{ell'}
\ell_{a,y}^{\prime}(1)=-y\mathbb{E}\left[\frac{Z^{\Tilde{\pi}^*}_{\tau}}{B_{\tau}}x_{h_a}^*(y\frac{Z^{\Tilde{\pi}^*}_{\tau}}{B_{\tau}})\right].
\end{eqnarray}
\end{lem}

\begin{proof}[Proof of Proposition \ref{unique}]
To prove the first claim, let us consider 
\begin{align*}
g(y):=\Tilde{V}_a(y)+xy,\quad a,x,y\in\mathbb{R}_+.
\end{align*}
Thanks to the differentiability of the function $\Phi_{h_a}$, we can apply the dominated convergence theorem and \eqref{Phi'} to obtain 
\begin{eqnarray}\label{g'}
g^{\prime}(y)=-\mathbb{E}\left[\frac{Z^{\Tilde{\pi}^*}_{\tau}}{B_{\tau}}x_{h_a}^*(y\frac{Z^{\Tilde{\pi}^*}_{\tau}}{B_{\tau}})\right]+x.
\end{eqnarray}
Inheriting from the function $x^*_{h_a}$, the function $g^{\prime}(y)$ is increasing with $\lim_{y\rightarrow0+}g^{\prime}(y)=-\infty$ and $\lim_{y\rightarrow\infty}g^{\prime}(y)=x>0$.
Hence, one gets the existence of a unique minimizer $y_x^*\in\mathbb{R}_+$ of $g(y)$.

For the second claim, 
on one hand, by \eqref{H(a)=inf}, one gets that
\begin{eqnarray}
\inf_{\lambda\in\mathbb{R}_+}\left\{\Tilde{J}_a(\lambda y_x^*,\Tilde{\pi}^*)+\lambda xy_x^*\right\}
\hspace{-0.3cm}&=&\hspace{-0.3cm}
\inf_{y\in\mathbb{R}_+}\left\{\Tilde{J}_a(y,\Tilde{\pi}^*)+xy\right\}
\nonumber\\
\hspace{-0.3cm}&\geq&\hspace{-0.3cm}
\inf_{y\in\mathbb{R}_+}\left\{\Tilde{V}_a(y)+xy\right\}=\Tilde{V}_a(y_x^*)+xy_x^*=\Tilde{J}(y_x^*,\Tilde{\pi}^*)+xy_x^*,\nonumber
\end{eqnarray}
which implies that the function $\lambda\mapsto\ell_{a,y_x^*}(\lambda)+\lambda xy_x^*$ attains its minimum at $\lambda=1$. As a result,  
\begin{eqnarray}
\ell_{a,y_x^*}^{\prime}(1)+xy_x^*=0.\nonumber
\end{eqnarray}
On the other hand, we assume that $y^*_x$ satisfies the equation \eqref{ell(1).y*}. Thanks to \eqref{Phi'} and \eqref{g'}, one obtains that
\begin{eqnarray}
\frac{\partial}{\partial y}\left[\Tilde{V}_a(y)+xy\right]\Big|_{y=y_x^*}=-\mathbb{E}\left[\frac{Z^{\Tilde{\pi}^*}_{\tau}}{B_{\tau}}x_{h_a}^*(y\frac{Z^{\Tilde{\pi}^*}_{\tau}}{B_{\tau}})\right]\bigg|_{y=y_x^*}+x=\left(\frac{1}{y}\ell_{a,y}^{\prime}(1)+x\right)\bigg|_{y=y_x^*}=0,\nonumber
\end{eqnarray}
which together with the fact that $g(y)=\Tilde{V}_a(y)+xy$ is strictly convex implies that $y_x^*\in\mathbb{R}_+$ is the minimizer of $\Tilde{V}_a(y)+xy$ on $\mathbb{R}_+$.
\end{proof}

\ \\
\begin{proof}[Proof of Proposition \ref{relationship}]
Using Lemma \ref{ell'(1).equ} and Proposition \ref{unique}, one knows that $\inf_{y\in\mathbb{R}_+}\{\Tilde{V}_a(y)+y\}=\Tilde{V}_a(y^*)+y^*$ with $y^*=y^*_x|_{x=1}$ being characterized in Proposition \ref{unique}.
It clearly holds that
\begin{eqnarray}\label{dX/B}
d\left[\frac{X_t}{B_t}\right]=\frac{\pi^{\mathsf{T}}_t}{B_t}\left[(\mu-r\mathbf{1})dt+\sigma dW_s\right].
\end{eqnarray}
In addition, by \eqref{dX/B} and \eqref{Z.pi}, for any $\Tilde{\pi}\in\Tilde{\mathcal{U}}_{0,\tau}$ and $\pi\in\mathcal{U}_{0,\tau}$, an application of It\^o's formula to the product of the process $Z^{\Tilde{\pi}}$ and $X/B$ yields
\begin{eqnarray}
Z^{\Tilde{\pi}}_t\frac{X_t}{B_t}+\int_0^t\frac{Z^{\Tilde{\pi}}_s}{B_s}\pi^{\mathsf{T}}_s\Tilde{\pi}_sds=x+\int_0^t\frac{Z^{\Tilde{\pi}}_s}{B_s}\left[\pi^{\mathsf{T}}_s\sigma-X_s\left[\xi+{\sigma}^{-1}{\Tilde{\pi}_s}\right]^{\mathsf{T}}\right]dW_s,\nonumber
\end{eqnarray}
which is a non-negative local martingale, hence a supermartingale. As a result, it holds that
\begin{eqnarray}
\mathbb{E}\left[Z^{\Tilde{\pi}}_\tau\frac{X_\tau}{B_\tau}+\int_0^\tau\frac{Z^{\Tilde{\pi}}_s}{B_s}\pi^{\mathsf{T}}_s\Tilde{\pi}_sds\right]\leq x,\nonumber
\end{eqnarray}
which implies that
\begin{eqnarray}\label{E<x.log}
\mathbb{E}\left[Z^{\Tilde{\pi}}_\tau\frac{X_\tau}{B_\tau}\right]\leq x.
\end{eqnarray}
Putting $y=y^*\frac{Z^{\Tilde{\pi}}_{\tau}}{B_{\tau}}$ in \eqref{phi(y)=h}, we have
$$h_a(X_{\tau})\leq\Phi_{h_a}\left(y^*\frac{Z^{\Tilde{\pi}_{\tau}}}{B_{\tau}}\right)+y^*\frac{Z^{\Tilde{\pi}_{\tau}}}{B_{\tau}}X_{\tau}.$$
Taking expectation on both sides in the above equation and using \eqref{E<x.log}, we obtain that 
\begin{eqnarray}\label{5.25}
\mathbb{E}[h_a(X_{\tau})]\leq \mathbb{E}\left[\Phi_{h_a}\left(y^*\frac{Z^{\Tilde{\pi}}_{\tau}}{B_{\tau}}\right)\right]+y^*\mathbb{E}\left[\frac{Z^{\Tilde{\pi}}_{\tau}}{B_{\tau}}X_{\tau}\right]\leq\Tilde{J}(y^*,\Tilde{\pi}^*)+xy^*=\Tilde{V}_a(y^*)+xy^*.
\end{eqnarray}
Hence, by the arbitrariness if $\pi$, we arrive at
\begin{eqnarray}\label{5.13}
\sup_{\pi\in\mathcal{U}_0(1)}\mathbb{E}[h_a(X_{\tau})]\leq \Tilde{V}_a(y^*)+xy^*.
\end{eqnarray}
If we take $X_{\tau}=X^*_{\tau}:=x^*_{h_a}\left(y^*\frac{Z^{\Tilde{\pi}^*}_{\tau}}{B_{\tau}}\right)$ in \eqref{5.25}, then by \eqref{phi(y)=h}, it holds that
$$\mathbb{E}[h_a(X_{\tau})]=\mathbb{E}\left[\Phi_{h_a}\left(y^*\frac{Z^{\Tilde{\pi}^*}_{\tau}}{B_{\tau}}\right)\right]+y^*\mathbb{E}\left[\frac{Z^{\Tilde{\pi}^*}_{\tau}}{B_{\tau}}X_{\tau}\right]=\Tilde{V}_a(y^*)+xy^*.$$
This, combining with \eqref{5.13}, implies that the first and second claim hold true.

To prove the third claim, define a martingale by $$M_t:=\mathbb{E}\left[\left.\frac{Z^{\Tilde{\pi}^*}_{\tau}}{B_{\tau}}x^*_{h_a}\left(y^*\frac{Z^{\Tilde{\pi}^*}_{\tau}}{B_{\tau}}\right)\right|\mathcal{F}_t\right],\quad t\in[0,\tau].$$
From proposition \ref{unique} and \eqref{ell'}, one knows that $y^*$ satisfies $\ell^{\prime}_{a,y^*}(1)+y^*=0$, which yields that
$$M_0=\mathbb{E}\left[\frac{Z^{\Tilde{\pi}^*}_{\tau}}{B_{\tau}}x^*_{h_a}\left(y^*\frac{Z^{\Tilde{\pi}^*}_{\tau}}{B_{\tau}}\right)\right]=-\frac{\ell^{\prime}_{a,y^*}(1)}{y^*}=1.$$ Hence, by martingale representation theorem, $M_t$ admits the expression that
\begin{eqnarray}
\label{mart.repre.theo.}
M_t=1+\int_0^t\eta^{\mathsf{T}}_sdW_s,\quad t\in[0,\tau],
\end{eqnarray}
where $\eta$ is a $\{\mathcal{F}_t\}$-progressively measurable process satisfying $\int_0^{\tau}\|\eta_t\|^2dt<\infty$ almost surely. Let us define $\hat{X}_t$ as \eqref{hat.X}. It suffices to verify that $\hat{X}_0=1$ and $\hat{X}_{\tau}=X^*_{\tau}$ as well as $d\hat{X}_t=[r\hat{X}_t+(\pi_t^*)^{\mathsf{T}}(\mu-r\mathbf{1})]dt+(\pi_t^*)^{\mathsf{T}}\sigma dW_t$. To this end, we observe that
\begin{eqnarray}
d\left[\frac{B_{t}}{Z^{\Tilde{\pi}^*}_t}\right]=\frac{B_{t}}{Z^{\Tilde{\pi}^*}_t}\left[\left(r+\|\xi+\sigma^{-1}{\Tilde{\pi}^*_t}\|^2\right)dt+(\xi+\sigma^{-1}{\Tilde{\pi}^*_t})^{\mathsf{T}}dW_t\right],\quad t\in[0,\tau].\nonumber
\end{eqnarray}
Therefore,
\begin{eqnarray}\label{dhat.X}
d\hat{X}_t=d\left[\frac{B_t}{Z^{\Tilde{\pi}^*}_t}M_t\right]
\hspace{-0.3cm}&=&\hspace{-0.3cm}
\hat{X}_t\left[\left(r+\|\xi+\sigma^{-1}{\Tilde{\pi}^*_t}\|^2\right)dt+(\xi+\sigma^{-1}{\Tilde{\pi}^*_t})^{\mathsf{T}}dW_t\right]
\nonumber\\
\hspace{-0.3cm}&&\hspace{-0.3cm}
+\eta^{\mathsf{T}}_t\frac{\hat{X}_t}{M_t}dW_t
+\eta^{\mathsf{T}}_t\frac{\hat{X}_t}{M_t}(\xi+\sigma^{-1}{\Tilde{\pi}^*_t})dt
\nonumber\\
\hspace{-0.3cm}&=&\hspace{-0.3cm}
r\hat{X}_tdt+(\pi_t^*)^{\mathsf{T}}(\mu-r\mathbf{1})dt+(\pi_t^*)^{\mathsf{T}}{\Tilde{\pi}^*_t}dt+(\pi^*_t)^{\mathsf{T}}\sigma dW_t.
\end{eqnarray}
Recall that the function $h_a$ satisfies \eqref{h_a'>h_a'} and $\Tilde{\pi}^*$ is the optimal solution to the dual problem \eqref{dual.pro}. Then, by virtue of \eqref{dhat.X} and the Strong Duality Theorem 4.8 in \cite{XS92(a)}, it holds that $(\pi^*_t)^{\mathsf{T}}\Tilde{\pi}^*_t=0$, which together with \eqref{dhat.X} implies the third claim.
\end{proof}

\ \\
\begin{proof}[Proof of Proposition \ref{fixppower}]
Since the proofs for the cases $\alpha\in(0,1)$ and $\alpha\in(-\infty,0)$ are similar, we will only present the proof for the situation where $\alpha\in(0,1)$.
Consider any $a_1,a_2>0$. By Proposition \ref{relationship}, there exists an optimizer $\pi^{*}$ and the resulting wealth process $X^{*}$ such that
\begin{eqnarray}
\Psi(a_1)
\hspace{-0.3cm}&=&\hspace{-0.3cm}
\alpha\sup_{\pi\in{\mathcal{U}}_{0}(1)}\mathbb{E}\left[e^{-\delta\tau}U(X_{\tau})+\frac{1}{\alpha}a_1e^{-\delta\tau}X_{\tau}^{\alpha(1-\gamma)}\right]
\nonumber\\
\hspace{-0.3cm}&=&\hspace{-0.3cm}
\alpha\mathbb{E}\left[e^{-\delta\tau}U(X_{\tau}^{*})+\frac{1}{\alpha}a_1e^{-\delta\tau}(X_{\tau}^{*})^{\alpha(1-\gamma)}\right].\nonumber
\end{eqnarray}
Then, it holds that
\begin{eqnarray}
\label{2.25}
\Psi(a_1)-\Psi(a_2)
\hspace{-0.3cm}&=&\hspace{-0.3cm}
\alpha\mathbb{E}\left[e^{-\delta\tau}U(X_{\tau}^{*})+\frac{1}{\alpha}a_1e^{-\delta\tau}(X_{\tau}^{*})^{\alpha(1-\gamma)}\right]-\alpha\sup_{\pi\in{\mathcal{U}}_{0}(1)}\mathbb{E}\left[e^{-\delta\tau}U(X_{\tau})+\frac{1}{\alpha}a_2e^{-\delta\tau}X_{\tau}^{\alpha(1-\gamma)}\right]
\nonumber\\
\hspace{-0.3cm}&\leq&\hspace{-0.3cm}
\alpha\mathbb{E}\left[e^{-\delta\tau}U(X_{\tau}^{*})+\frac{1}{\alpha}a_1e^{-\delta\tau}(X_{\tau}^{*})^{\alpha(1-\gamma)}\right]-\alpha\mathbb{E}\left[e^{-\delta\tau}U(X_{\tau}^{*})+\frac{1}{\alpha}a_2e^{-\delta\tau}(X_{\tau}^{*})^{\alpha(1-\gamma)}\right]
\nonumber\\
\hspace{-0.3cm}&\leq&\hspace{-0.3cm}
(a_1-a_2)e^{-\delta\tau}\sup_{\pi\in{\mathcal{U}}_{0}(1)}\mathbb{E}[X_{\tau}^{\alpha(1-\gamma)}]
\nonumber\\
\hspace{-0.3cm}&=&\hspace{-0.3cm}
(a_1-a_2)e^{-(\delta-\zeta(\alpha(1-\gamma)))\tau}
\nonumber\\
\hspace{-0.3cm}&\leq&\hspace{-0.3cm}
e^{-(\delta-\zeta(\alpha(1-\gamma)))\tau}|a_1-a_2|.
\end{eqnarray}
Here, in the second equality we have used the fact that
$$\sup_{\pi\in{\mathcal{U}}_{0}(1)}\mathbb{E}[X_{\tau}^{\alpha(1-\gamma)}]=e^{\zeta(\alpha(1-\gamma))\tau},$$
because it is the value function of the utility maximization on terminal wealth under the prohibition of short-selling the maturity $\tau$ and the initial wealth of 1. Due to \eqref{2.25} and the standing Assumption \ref{ass1}, the existence of the unique fixed-point $A^*$ for $\Psi$ immediately follows from the Banach contraction theorem.

Next, noting that $\pi\equiv0$ is admissible to problem \eqref{problem3}, we readily obtain that
\begin{eqnarray}\label{lower}
\Psi(a)\geq e^{(r\alpha-\delta)\tau}+ae^{-(\delta-r\alpha(1-\gamma))\tau}.
\end{eqnarray}
On the other hand, we have
\begin{eqnarray}\label{upper}
\Psi(a)
\hspace{-0.3cm}&=&\hspace{-0.3cm}
\alpha\sup_{\pi\in{\mathcal{U}}_{0}(1)}\mathbb{E}\left[\frac{1}{\alpha}e^{-\delta\tau}X_{\tau}^{\alpha}+\frac{1}{\alpha}ae^{-\delta\tau}X_{\tau}^{\alpha(1-\gamma)}\right]
\nonumber\\
\hspace{-0.3cm}&\leq&\hspace{-0.3cm}
\sup_{\pi\in{\mathcal{U}}_{0}(1)}\mathbb{E}\left[e^{-\delta\tau}X_{\tau}^{\alpha}\right]+a\sup_{\pi\in{\mathcal{U}}_{0}(1)}\mathbb{E}\left[e^{-\delta\tau}X_{\tau}^{\alpha(1-\gamma)}\right]
\nonumber\\
\hspace{-0.3cm}&=&\hspace{-0.3cm}
e^{(\zeta(\alpha)-\delta)\tau}+ae^{-(\delta-\zeta(\alpha(1-\gamma)))\tau}.
\end{eqnarray}
In view of \eqref{lower}, it is easy to see that
$$A^*=\Psi(A^*)\geq e^{(r\alpha-\delta)\tau}+A^*e^{-(\delta-r\alpha(1-\gamma))\tau},$$
which gives the lower bound of $A^*$.

By \eqref{upper}, one can also derive that
$$A^*=\Psi(A^*)\leq e^{(\zeta(\alpha)-\delta)\tau}+A^*e^{-(\delta-\zeta(\alpha(1-\gamma)))\tau},$$
which leads to the upper bound of $A^*$. The proof is complete.
\end{proof}

\ \\
\begin{proof}[Proof of Theorem \ref{thm4.1}]
For any admissible process $\pi\in\mathcal{U}_0(x)$ and the associated wealth process $X$, define a discrete-time stochastic process $D=(D_n)_{n\geq 0}$ by
$$D_n:=\sum_{i=1}^ne^{-\delta T_i}U\bigg(\frac{X_{T_i}}{X_{T_{i-1}}^{\gamma}}\bigg)+\frac{1}{\alpha}A^*e^{-\delta T_n}X_{T_n}^{\alpha(1-\gamma)}.$$
Then
\begin{eqnarray}\label{Dn+1}
D_{n+1}
\hspace{-0.3cm}&=&\hspace{-0.3cm}
D_{n}+e^{-\delta T_{n+1}}U\bigg(\frac{X_{T_{n+1}}}{X_{T_{n}}^{\gamma}}\bigg)-\frac{1}{\alpha}A^*e^{-\delta T_n}X_{T_n}^{\alpha(1-\gamma)}+\frac{1}{\alpha}A^*e^{-\delta T_{n+1}}X_{T_{n+1}}^{\alpha(1-\gamma)}
\nonumber\\
\hspace{-0.3cm}&=&\hspace{-0.3cm}
D_n+e^{-\delta T_n}\left[e^{-\delta\tau}\left(U\bigg(\frac{X_{T_{n+1}}}{X^{\gamma}_{T_{n}}}\bigg)+\frac{1}{\alpha}A^*X^{\alpha(1-\gamma)}_{T_{n+1}}\right)-\frac{1}{\alpha}A^*X^{\alpha(1-\gamma)}_{T_n}\right].
\end{eqnarray}
Taking the conditional expectation on both sides of \eqref{Dn+1}, one gets
\begin{eqnarray}\label{3.4}
\mathbb{E}[D_{n+1}|\mathcal{F}_{T_n}]=D_n+e^{-\delta T_n}X^{\alpha(1-\gamma)}_{T_n}\left[e^{-\delta\tau}\mathbb{E}\left[\left.U\bigg(\frac{X_{T_{n+1}}}{X_{T_n}}\bigg)+\frac{1}{\alpha}A^*\bigg(\frac{X_{T_{n+1}}}{X_{T_n}}\bigg)^{\alpha(1-\gamma)}\right|\mathcal{F}_{T_n}\right]-\frac{1}{\alpha}A^*\right].
\end{eqnarray}
Furthermore, it holds that
\begin{eqnarray}\label{3.5}
\hspace{-0.3cm}&&\hspace{-0.3cm}\mathbb{E}\left[U\bigg(\frac{X_{T_{n+1}}}{X_{T_n}}\bigg)+\frac{1}{\alpha}A^*\bigg(\frac{X_{T_{n+1}}}{X_{T_n}}\bigg)^{\alpha(1-\gamma)}\bigg|\mathcal{F}_{T_n}\right]
\nonumber\\
\hspace{-0.3cm}&\leq&\hspace{-0.3cm}\sup_{\pi\in{\mathcal{U}}_{T_n}(1)}\mathbb{E}\left[U(X_{T_{n+1}})+\frac{1}{\alpha}A^*X_{T_{n+1}}^{\alpha(1-\gamma)}\bigg|\mathcal{F}_{T_n}\right]=\frac{1}{\alpha}H(A^*).
\end{eqnarray}
Recall that $A^*$ is the fixed-point of $\Psi(a)=e^{-\delta\tau}H(a)$. This, together with \eqref{3.4} and \eqref{3.5}, implies
\begin{eqnarray}\label{3.6}
\mathbb{E}[D_{n+1}|\mathcal{F}_{T_n}]\leq D_n+e^{-\delta T_n}X_{T_n}^{\alpha(1-\gamma)}\left[e^{-\delta\tau}\frac{1}{\alpha}H(A^*)-\frac{1}{\alpha}A^*\right]=D_n.
\end{eqnarray}
Hence, $(D_n)_{n\geq 0}$ is a $\{\mathcal{F}_{T_n}\}$-supermartingale, which yields that
\begin{eqnarray}
\frac{1}{\alpha}A^*x^{\alpha(1-\gamma)}=D_0\geq\mathbb{E}\left[\sum_{i=1}^ne^{-\delta T_i}U\bigg(\frac{X_{T_i}}{X_{T_{i-1}}^{\gamma}}\bigg)+\frac{1}{\alpha}A^*e^{-\delta T_n}X_{T_n}^{\alpha(1-\gamma)}\right].\nonumber
\end{eqnarray}
As a result, 
\begin{eqnarray}
\mathbb{E}\left[\sum_{i=1}^ne^{-\delta T_i}U\bigg(\frac{X_{T_i}}{X_{T_{i-1}}^{\gamma}}\bigg)\right]
\hspace{-0.3cm}&\leq&\hspace{-0.3cm}
\frac{1}{\alpha}A^*x^{\alpha(1-\gamma)}+e^{-\delta T_n}A^*
\sup_{X\in\mathcal{U}_0(x)}
\mathbb{E}\left[\frac{1}{\alpha}X^{\alpha(1-\gamma)}_{T_n}\right]
\nonumber\\
\hspace{-0.3cm}&=&\hspace{-0.3cm}
\frac{1}{\alpha}A^*x^{\alpha(1-\gamma)}+\frac{1}{\alpha}e^{-(\delta-\zeta(\alpha(1-\gamma))) T_n}A^*x^{\alpha(1-\gamma)},\nonumber
\end{eqnarray}
where we have used again the fact that $\sup_{X\in\mathcal{U}_0(x)}\mathbb{E}\big[\frac{1}{\alpha}X^{\alpha(1-\gamma)}_{T_n}\big]=\frac{1}{\alpha}x^{\alpha(1-\gamma)}e^{\zeta(\alpha(1-\gamma)) T_n}$. By Assumption \ref{ass1} (i.e., $\delta>\zeta(\alpha(1-\gamma))\vee0$) and the monotone convergence theorem, we have
\begin{eqnarray}
\mathbb{E}\left[\sum_{i=1}^{\infty}e^{-\delta T_i}U\bigg(\frac{X_{T_i}}{X_{T_{i-1}}^{\gamma}}\bigg)\right]\leq \frac{1}{\alpha}A^*x^{\alpha(1-\gamma)},\nonumber
\end{eqnarray}
and hence
\begin{eqnarray}
V(x)=\sup_{\pi\in\mathcal{U}_0(x)}\mathbb{E}\left[\sum_{i=1}^{\infty}e^{-\delta T_i}U\bigg(\frac{X_{T_i}}{X_{T_{i-1}}^{\gamma}}\bigg)\right]\leq \frac{1}{\alpha}A^*x^{\alpha(1-\gamma)}.\nonumber
\end{eqnarray}
For the reverse inequality, it suffices to show the existence of some admissible process $X^*$ such that $$\mathbb{E}\left[\sum_{i=1}^{\infty}e^{-\delta T_i}U\left(\frac{X^*_{T_i}}{(X^*_{T_{i-1}})^{\gamma}}\right)\right]=\frac{1}{\alpha}A^*x^{\alpha(1-\gamma)}.$$
By Proposition \ref{relationship}, under the choice of $\frac{X_{T_{n+1}}}{X_{T_n}}=x^*_{h_{A^*}}\left(y^*_n\frac{Z^{\Tilde{\pi}^*}_{T_{n+1}}/B_{T_{n+1}}}{Z^{\Tilde{\pi}^*}_{T_n}/B_{T_n}}\right)$, it holds that
\begin{eqnarray}
\hspace{-0.3cm}&&\hspace{-0.3cm}
\mathbb{E}\left[\left.U\bigg(\frac{X_{T_{n+1}}}{X_{T_n}}\bigg)+\frac{1}{\alpha}A^*\bigg(\frac{X_{T_{n+1}}}{X_{T_n}}\bigg)^{\alpha(1-\gamma)}\right|\mathcal{F}_{T_n}\right]
\nonumber\\
\hspace{-0.3cm}&=&\hspace{-0.3cm}
\sup_{\pi\in{\mathcal{U}}_{T_n}(1)}\mathbb{E}\left[U(X_{T_{n+1}})+\frac{1}{\alpha}A^*X_{T_{n+1}}^{\alpha(1-\gamma)}\bigg|\mathcal{F}_{T_n}\right]=\alpha H(A^*).\nonumber
\end{eqnarray}
Here, $y^*_n$ is the unique solution to the equation
\begin{eqnarray}
\ell^{(n)\,\prime}_{A^*,y}(1)+y=0,\nonumber
\end{eqnarray}
where the function $\mathbb{R}_+\ni\lambda\mapsto\ell_{A^*,y}^{(n)}(\lambda)\in[0,\infty)$ is given by 
\begin{align*}
\ell^{(n)}_{A^*,y}(\lambda) &:=
\mathbb{E}\left[\left.\Phi_{h_{A^*}}\left(\lambda y\frac{Z^{\Tilde{\pi}^*}_{T_{n+1}}/B_{T_{n+1}}}{Z^{\Tilde{\pi}^*}_{T_{n}}/B_{T_{n}}}\right)\right|\mathcal{F}_{T_{n}}\right]
\nonumber\\
&=\mathbb{E}\left[\left.h_{A^*}\left(x^*_{h_{A^*}}\left(\lambda y\frac{Z^{\Tilde{\pi}^*}_{T_{n+1}}/B_{T_{n+1}}}{Z^{\Tilde{\pi}^*}_{T_{n}}/B_{T_{n}}}\right)\right)-x^*_{h_{A^*}}\left(\lambda y\frac{Z^{\Tilde{\pi}^*}_{T_{n+1}}/B_{T_{n+1}}}{Z^{\Tilde{\pi}^*}_{T_{n}}/B_{T_{n}}}\right)\lambda y\frac{Z^{\Tilde{\pi}^*}_{T_{n+1}}/B_{T_{n+1}}}{Z^{\Tilde{\pi}^*}_{T_{n}}/B_{T_{n}}}\right|\mathcal{F}_{T_{n}}\right].
\end{align*}
Furthermore, by Lemma \ref{ell'(1).equ}, one has that 
\begin{eqnarray}
\ell_{a,y}^{(n)\,\prime}(1)=-y\mathbb{E}\left[\left.\frac{Z^{\Tilde{\pi}^*}_{T_{n+1}}/B_{T_{n+1}}}{Z^{\Tilde{\pi}^*}_{T_{n}}/B_{T_{n}}}x_{h_a}^*\left(y\frac{Z^{\Tilde{\pi}^*}_{T_{n+1}}/B_{T_{n+1}}}{Z^{\Tilde{\pi}^*}_{T_{n}}/B_{T_{n}}}\right)\right|\mathcal{F}_{T_{n}}\right].\nonumber
\end{eqnarray}
In view of $Z^{\Tilde{\pi}^*}$ in \eqref{Z.tilde} and the stationary increment property of Brownian motion, we can derive that
\begin{eqnarray}\label{3.11}
\ell_{A^*,y^*_n}^{(n)\,\prime}(1)
\hspace{-0.3cm}&=&\hspace{-0.3cm}
-\mathbb{E}\left[y_n^*\frac{Z^{\Tilde{\pi}^*}_{T_{n+1}}/B_{T_{n+1}}}{Z^{\Tilde{\pi}^*}_{T_{n}}/B_{T_{n}}}x_{h_{A^*}}^*\left(y_n^*\frac{Z^{\Tilde{\pi}^*}_{T_{n+1}}/B_{T_{n+1}}}{Z^{\Tilde{\pi}^*}_{T_{n}}/B_{T_{n}}}\right)\Bigg|\mathcal{F}_{T_n}\right]
\nonumber\\
\hspace{-0.3cm}&=&\hspace{-0.3cm}
-\mathbb{E}\left[y_n^*\frac{Z^{\Tilde{\pi}^*}_{T_1}}{B_{T_1}}x_{h_{A^*}}^*\left(y_n^*\frac{Z^{\Tilde{\pi}^*}_{T_1}}{B_{T_1}}\right)\Bigg|\mathcal{F}_{T_0}\right]
=
-\mathbb{E}\left[y^*\frac{Z^{\Tilde{\pi}^*}_{\tau}}{B_{\tau}}x_{h_{A^*}}^*\left(y^*\frac{Z^{\Tilde{\pi}^*}_{\tau}}{B_{\tau}}\right)\right]=-y^*,
\end{eqnarray}
and $y^*_n=y^*_1=y^*$ for all $n\geq 1$.

Next, let us define a sequence of random variables recursively as follows
\begin{eqnarray}
Q_n:=\left\{
\begin{array}{ll}
xx^*_{h_{A^*}}\left(y^*\frac{Z^{\Tilde{\pi}^*}_{\tau}}{B_{\tau}}\right), & n=1, \\
     Q_{n-1}x^*_{h_{A^*}}\left(y^*\frac{Z^{\Tilde{\pi}^*}_{T_i}/B_{T_{i}}}{Z^{\Tilde{\pi}^*}_{T_{i-1}}/B_{T_{i-1}}}\right),& n=2,3,....
\end{array}
\right.\nonumber
\end{eqnarray}
Then, by \eqref{3.11}, it holds that
\begin{eqnarray}
\mathbb{E}\left[\left.\frac{Z^{\Tilde{\pi}^*}_{T_{n+1}}}{B_{T_{n+1}}}Q_{{n+1}}\right|\mathcal{F}_{T_n}\right]=\mathbb{E}\left[\left.\frac{Z^{\Tilde{\pi}^*}_{T_{n+1}}}{B_{T_{n+1}}}Q_{n}x^*_{h_{A^*}}\left(y^*\frac{Z^{\Tilde{\pi}^*}_{T_{n+1}}/B_{T_{n+1}}}{Z^{\Tilde{\pi}^*}_{T_{n}}/B_{T_{n}}}\right)\right|\mathcal{F}_{T_n}\right]=\frac{Z^{\Tilde{\pi}^*}_{T_n}}{B_{T_n}}Q_{n},\nonumber
\end{eqnarray}
which implies that $\big(\frac{Z^{\Tilde{\pi}^*}_{T_n}}{B_{T_n}}Q_{n}\big)_{n\geq 1}$ is a $\{\mathcal{F}_{T_n}\}$-martingale. Using standard arguments of martingale representation theorem, there exists a $\pi^*\in\mathcal{U}_0(x)$ such that the associated wealth process $X^*_{T_n}=Q_n$ for all $n$. Define
$$D^*_n:=\sum_{i=1}^ne^{-\delta T_i}U\bigg(\frac{X^*_{T_i}}{{(X^*_{T_{i-1}})}^{\gamma}}\bigg)+\frac{1}{\alpha}A^*e^{-\delta T_n}{(X^*_{T_n})}^{\alpha(1-\gamma)}, \quad n\geq 0.$$
Using the same arguments leading to \eqref{3.6}, it is easy to conclude that
\begin{eqnarray}
\mathbb{E}[D^*_{n+1}|\mathcal{F}_{T_n}]=D^*_n+e^{-\delta T_n}{(X^*_{T_n})}^{\alpha(1-\gamma)}\left[e^{-\delta\tau}\frac{1}{\alpha}H(A^*)-\frac{1}{\alpha}A^*\right]=D^*_n, \quad n\geq 0,\nonumber
\end{eqnarray}
which yields that $(D^*_n)_{n\geq0}$ is a $\{\mathcal{F}_{T_n}\}$-martingale. Hence, we have
\begin{eqnarray}
\mathbb{E}\left[\sum_{i=1}^{n}e^{-\delta T_i}U\bigg(\frac{X^*_{T_i}}{(X^*_{T_{i-1}})^{\gamma}}\bigg)\right]
\hspace{-0.3cm}&=&\hspace{-0.3cm}
\frac{1}{\alpha}A^*x^{\alpha(1-\gamma)}-\frac{1}{\alpha}A^*\mathbb{E}\left[e^{-\delta T_n}(X^*_{T_n})^{\alpha(1-\gamma)}\right].\nonumber
\end{eqnarray}
By Assumption \ref{ass1} and the monotone convergence theorem, we have
\begin{eqnarray}
\mathbb{E}\left[\sum_{i=1}^{\infty}e^{-\delta T_i}U\bigg(\frac{X^*_{T_i}}{(X^*_{T_{i-1}})^{\gamma}}\bigg)\right]=\frac{1}{\alpha}A^*x^{\alpha(1-\gamma)},\quad x\in\mathbb{R}_+.\nonumber
\end{eqnarray}
Finally, the representation of the optimal wealth process $X^*$ and the optimal portfolio process $\pi^*$ defined by \eqref{X.rep} and \eqref{pi.rep} can be verify by applying similar arguments used in the proof of Proposition \ref{relationship}, and hence, we omit it.
The proof is complete.
\end{proof}

\ \\
\begin{proof}[Proof of Corollary \ref{coro4.1}]
In view of $\gamma=1$ and Theorem \ref{thm4.1}, one knows that the value function of the problem \eqref{problem} is given by $V(x)=\frac{1}{\alpha}A^*$ for all $x\in\mathbb{R}_+$ with $A^*$ being the unique fixed-point of the function $\Psi$ defined in \eqref{Psi}. More specifically, we have
$$A^*=e^{-\delta\tau}H(A^*)=e^{-\delta\tau}\alpha\sup_{\pi\in\mathcal{U}_{0}(1)}\mathbb{E}\left[\frac{1}{\alpha}X_{\tau}^{\alpha}+\frac{1}{\alpha}A^*\right]=e^{-\delta\tau}A^*+e^{-\delta\tau}e^{\zeta(\alpha)\tau},$$
which implies that
$A^*=\frac{e^{(\zeta(\alpha)-\delta)\tau}}{1-e^{-\delta\tau}}$. The first claim holds.

For the second claim, when $\gamma=1$, the function $h_{A^*}$ reduces to $h_{A^*}(x)=\frac{1}{\alpha}x^{\alpha}$ with its inverse function being $x^*_{h_{A^*}}(y)=y^{\frac{1}{\alpha-1}}$. Then, we solve the equation $\ell^{\prime}_{A^*,y}(1)+y=0$ and derive its solution $$y^*=\left(\mathbb{E}\left[\left(\frac{Z^{\Tilde{\pi}^*}_{\tau}}{B_{\tau}}\right)^{\frac{\alpha}{\alpha-1}}\right]\right)^{1-\alpha}=e^{\zeta(\alpha)\tau}.$$
Using Theorem \ref{thm4.1}, we have
\begin{eqnarray}
X^*_{T_i}=X^*_{T_{i-1}}e^{\frac{\zeta(\alpha)}{\alpha-1}\tau}\left(\frac{Z^{\Tilde{\pi}}_{T_i}/B_{T_i}}{Z^{\Tilde{\pi}}_{T_{i-1}}/B_{T_{i-1}}}\right)^{\frac{1}{\alpha-1}}.\nonumber
\end{eqnarray}
Furthermore, one can use the results of Section 3 and Section 4 in \cite{XS92(b)} to obtain that the optimal wealth process $X^*$
for one-period terminal wealth utility maximization problem $\sup_{\pi\in\mathcal{U}_0(1)}\mathbb{E}[\frac{1}{\alpha}X_{\tau}^{\alpha}]$ can be expressed as
\begin{eqnarray}
X^*_{t}=xe^{\frac{\zeta(\alpha)}{\alpha-1}\tau}\left(\frac{Z^{\Tilde{\pi}}_{t}}{B_{t}}\right)^{\frac{1}{\alpha-1}},\quad t\in[0,\tau].\nonumber
\end{eqnarray}
and 
\begin{eqnarray}
\pi^*_t=\frac{1}{1-\alpha}X^*_t(\sigma^{\mathsf{T}})^{-1}[\xi+\sigma^{-1}\Tilde{\pi}^*],\quad t\in[0,\tau].\nonumber
\end{eqnarray}
Then, by applying similar argument as used in the proof of Theorem \ref{thm4.1}, we arrive at the desired results.
\end{proof}

\ \\
\begin{proof}[Proof of Proposition \ref{taupower}]
Following from Corollary \ref{coro4.1}, when $\gamma=1$, the value function of the problem \eqref{problem} is given by 
$$V(x;\tau)=\frac{1}{\alpha}A^*=\frac{e^{(\zeta(\alpha)-\delta)\tau}}{\alpha(1-e^{-\delta\tau})},\quad (x,\tau)\in\mathbb{R}_+\times\mathbb{R}_+.$$ Let 
\begin{eqnarray}\label{g(tau)}
g(\tau):=A^*\tau=\frac{e^{(\zeta(\alpha)-\delta)\tau}\tau}{1-e^{-\delta\tau}},\quad \tau\in\mathbb{R}_+.
\end{eqnarray}
Differentiating both sides of \eqref{g(tau)} gives
\begin{eqnarray}
g^{\prime}(\tau)=
\frac{e^{(\zeta(\alpha)-\delta)\tau}[(1-e^{-\delta\tau})(\zeta(\alpha)\tau+1)-\delta\tau]}{(1-e^{-\delta\tau})^2},\quad\tau\in\mathbb{R}_+.\nonumber
\end{eqnarray}
One can verify that, under the assumption $\frac{\delta}{2}<\zeta(\alpha)<\delta$, it holds that
$$\lim_{\tau\rightarrow0+}g(\tau)=\frac{1}{\delta}>0,\quad
\lim_{\tau\rightarrow\infty}g(\tau)=0,\quad \text{and} \quad \lim_{\tau\rightarrow0+}g^{\prime}(\tau)=\frac{2\zeta(\alpha)-\delta}{2\delta}>0.$$
It therefore follows that $g(\tau)$ attains its maximum at $\tau^*\in(0,\infty)$, which completes the proof.
\end{proof}
\subsection{Proofs in Section \ref{section3}}
\begin{proof}[Proof of Proposition \ref{existence.log}] 
We again apply an adapted version of Theorem 2.1 of \cite{XS92(b)} to explicitly characterize an optimal solution $\Tilde{\pi}_y\in\Tilde{\mathcal{U}}_{0,\tau}$ to the dual problem \eqref{dual.pro.log} with initial condition $y$. Actually, instead of checking $\mathbb{E}\left[-\log(Z_{t}^{\tilde{\pi}})\right]<\infty$ (the same type of condition is used in Lemma 5.1 of \cite{XS92(b)}), we can simply check the following stronger condition as
\begin{align}
    \hspace{-0.3cm}
    \mathbb{E}\left[\left|-\log(Z_{t}^{\tilde{\pi}})\right|\right]
    &=\mathbb{E}\left[\left|\int_0^t\left[\xi+{\sigma}^{-1}{\Tilde{\pi}_s}\right]^{\mathsf{T}}dW_s+\frac{1}{2}\int_0^t\left\|\xi+{\sigma}^{-1}{\Tilde{\pi}_s}\right\|^2ds\right|\right]
    \nonumber\\
    &\leq \left[\mathbb{E}\left[\left(\int_0^t\left[\xi+{\sigma}^{-1}{\Tilde{\pi}_s}\right]^{\mathsf{T}}dW_s\right)^{2}\right]\right]^{1/2}+\frac{1}{2}\mathbb{E}\left[\int_0^t\left\|\xi+{\sigma}^{-1}{\Tilde{\pi}_s}\right\|^2ds\right]
    \nonumber\\
    &= \left[\mathbb{E}\left[\int_0^t\left\|\xi+{\sigma}^{-1}{\Tilde{\pi}_s}\right\|^{2} ds\right]\right]^{1/2}+\frac{1}{2}\mathbb{E}\left[\int_0^t\left\|\xi+{\sigma}^{-1}{\Tilde{\pi}_s}\right\|^2ds\right]<\infty,\quad t\geq 0,
    \label{checkcond.}
\end{align}
where, in the last inequality, we have taken into account the facts that $\left\|\xi+{\sigma}^{-1}{\Tilde{\pi}_s}\right\|^2\leq 2\left\|\xi\right\|^2+2\left\|{\sigma}^{-1}{\Tilde{\pi}_s}\right\|^2\leq 2\left\|\xi\right\|^2+2\lambda_{\max}\left\|{\Tilde{\pi}_s}\right\|^2$ ($\lambda_{\max}>0$ denotes the largest eigenvalue of the positive-definite matrix $(\sigma^{-1})^{\mathsf{T}}\sigma^{-1}$) and $\mathbb{E}\left[\int_0^t\|\Tilde{\pi}_u\|^2du\right]<\infty$ (see \eqref{tilde.U}). Combining \eqref{checkcond.}, the non-increasing and convex property of the function $(0,\infty)\ni x\mapsto -\log x$ (note that we do not need the lower-boundedness of this function as Lemma 5.1 of \cite{XS92(b)} does), one can follow similar arguments of Lemma 5.1 of \cite{XS92(b)} to prove the submartingale property of $(-\log(Z_{t}^{\tilde{\pi}}))_{t\geq 0}$.
Then, we can follow the same arguments of Theorems 2.1 and 5.2 in \cite{XS92(b)} to verify that the unique minimizer $\Tilde{\pi}^*$ of the function $[0,\infty)^n \ni \Tilde{\pi}\mapsto\|\xi+\sigma^{-1}\Tilde{\pi}\|$ is an optimal solution to the dual problem \eqref{dual.pro.log} with initial condition $y$.

Moreover, one can apply a similar argument as used in the proofs of Proposition \ref{unique} and \ref{relationship} to verify the second and third claims hold true, and hence, we omit the proof.
\end{proof}

\begin{lem}
\label{lem3.1}
Suppose that
\begin{align}
\label{standing ass.}
    \sum_{i=1}^{\infty}e^{-\delta T_i}\mathbb{E}\left[\left(\log\left(\frac{X_{T_i}}{\left(X_{T_{i-1}}\right)^{\gamma}}\right)\right)_{-}\right]<\infty.
\end{align} We have
    \begin{align}
    \label{|logX|up}
        \mathbb{E}\left[\left|\log X_{T_n}\right|\right]\leq \left|\log x\right|+n\left(\frac{1}{2}\|\Tilde{\xi}\|^2+r\right)\tau
+2\sum_{i=1}^{n}\mathbb{E}\left[\left(\log \frac{X_{T_i}}{(X_{T_{i-1}})^{\gamma}}\right)_{-}\right], \quad n\geq 1,
    \end{align}
and
    \begin{align}
    \label{lim=0}
\lim\limits_{n\rightarrow\infty}e^{-\delta T_{n}}\mathbb{E}\left[\left|\log X_{T_n}\right|\right]=0.
    \end{align}
\end{lem}

\begin{proof} By \eqref{Vtilde} and \eqref{3.15}, we know that
\begin{align}
\label{ElogXup}
    \mathbb{E}\left[\log X_{T_n}\right]\leq \log x+n\left(\frac{1}{2}\|\Tilde{\xi}\|^2+r\right)\tau.
\end{align}
Using the facts that $(a+b)_{-}\leq a_{-}+b_{-}$ and $\gamma\in(0,1]$, one gets that
\begin{align}
\label{3.23.v}
    \mathbb{E}\left[\left(\log X_{T_n}\right)_{-}\right]
    &=\mathbb{E}\left[\left(\log \frac{X_{T_n}}{(X_{T_{n-1}})^{\gamma}}+\gamma\log X_{T_{n-1}}\right)_{-}\right]\nonumber\\
    &\leq \mathbb{E}\left[\left(\log \frac{X_{T_n}}{(X_{T_{n-1}})^{\gamma}}\right)_{-}+\gamma\left(\log X_{T_{n-1}}\right)_{-}\right]
    \nonumber\\
    &\leq \mathbb{E}\left[\left(\log \frac{X_{T_n}}{(X_{T_{n-1}})^{\gamma}}\right)_{-}\right]+\mathbb{E}\left[\left(\log X_{T_{n-1}}\right)_{-}\right].
\end{align}
Repeating the above arguments of \eqref{3.23.v} $n$ times yields that
\begin{align}
\label{logX-up}
    \mathbb{E}\left[\left(\log X_{T_n}\right)_{-}\right]
    &\leq \sum_{i=1}^{n}\mathbb{E}\left[\left(\log \frac{X_{T_i}}{(X_{T_{i-1}})^{\gamma}}\right)_{-}\right]+\left(\log x\right)_{-}.
\end{align}
Combining \eqref{ElogXup} and \eqref{logX-up} leads to 
\begin{align}
\label{logX+up}
    \mathbb{E}\left[\left(\log X_{T_n}\right)_{+}\right]
    &=\mathbb{E}\left[\log X_{T_n}\right]
    +\mathbb{E}\left[\left(\log X_{T_n}\right)_{-}\right]
    \nonumber\\
    &\leq 
    \log x+n\left(\frac{1}{2}\|\Tilde{\xi}\|^2+r\right)\tau
+\sum_{i=1}^{n}\mathbb{E}\left[\left(\log \frac{X_{T_i}}{(X_{T_{i-1}})^{\gamma}}\right)_{-}\right]+\left(\log x\right)_{-}
\nonumber\\
    &=\left(\log x\right)_{+}+n\left(\frac{1}{2}\|\Tilde{\xi}\|^2+r\right)\tau
+\sum_{i=1}^{n}\mathbb{E}\left[\left(\log \frac{X_{T_i}}{(X_{T_{i-1}})^{\gamma}}\right)_{-}\right].
\end{align}
Putting together \eqref{logX-up} and \eqref{logX+up} gives \eqref{|logX|up}. In addition, \eqref{lim=0} is a direct consequence of \eqref{standing ass.} and \eqref{|logX|up}.
\end{proof}

\ \\
\begin{proof}[Proof of Theorem \ref{thm3.1}]
For any admissible portfolio process $\pi\in\mathcal{U}_0(x)$ and the resulting wealth process $X$, we define a discrete-time stochastic process $D=(D_n)_{n=1,2,...}$ by
$$D_n:=\sum_{i=1}^ne^{-\delta T_i}\log \bigg(\frac{X_{T_i}}{X_{T_{i-1}}^{\gamma}}\bigg)+e^{-\delta T_n}(A^*+C^*\log X_{T_n}),$$
with $D_0:=A^*+C^*\log x$. 
Similar to Theorem \ref{thm4.1}, one has
$(D_n)_{n\geq 0}$ is a $\{\mathcal{F}_{T_n}\}$-supermartingale and 
\begin{align}
\label{supmarleadsto}
&A^*+C^*\log x=D_0
\nonumber\\
\geq
&\mathbb{E}\left[\sum_{i=1}^ne^{-\delta T_i}\log\bigg(\frac{X_{T_i}}{X_{T_{i-1}}^{\gamma}}\bigg)+e^{-\delta T_n}(A^*+C^*\log X_{T_n})\right]
\nonumber\\
=&\mathbb{E}\left[\sum_{i=1}^ne^{-\delta T_i}\left(\log\bigg(\frac{X_{T_i}}{X_{T_{i-1}}^{\gamma}}\bigg)\right)_{+}-\sum_{i=1}^ne^{-\delta T_i}\left(\log\bigg(\frac{X_{T_i}}{X_{T_{i-1}}^{\gamma}}\bigg)\right)_{-}+e^{-\delta T_n}(A^*+C^*\log X_{T_n})\right].
\end{align}
Thanks to \eqref{standing ass.}, it holds that
\begin{align}
    \mathbb{E}\left[\sum_{i=1}^ne^{-\delta T_i}\left(\log\bigg(\frac{X_{T_i}}{X_{T_{i-1}}^{\gamma}}\bigg)\right)_{-}\right]
    &=\sum_{i=1}^ne^{-\delta T_i}\mathbb{E}\left[\left(\log\bigg(\frac{X_{T_i}}{X_{T_{i-1}}^{\gamma}}\bigg)\right)_{-}\right]
    \nonumber\\
    &\leq \sum_{i=1}^{\infty}
    e^{-\delta T_i}\mathbb{E}\left[\left(\log\bigg(\frac{X_{T_i}}{X_{T_{i-1}}^{\gamma}}\bigg)\right)_{-}\right]<\infty,\quad n\geq 1.\nonumber
\end{align}
Hence, one can rewrite \eqref{supmarleadsto} as
\begin{align*}
A^*+C^*\log x
\geq& \mathbb{E}\left[\sum_{i=1}^ne^{-\delta T_i}\left(\log\bigg(\frac{X_{T_i}}{X_{T_{i-1}}^{\gamma}}\bigg)\right)_{+}\right]-\mathbb{E}\left[\sum_{i=1}^ne^{-\delta T_i}\left(\log\bigg(\frac{X_{T_i}}{X_{T_{i-1}}^{\gamma}}\bigg)\right)_{-}\right]
\nonumber\\
 &+
\mathbb{E}\left[e^{-\delta T_n}(A^*+C^*\log X_{T_n})\right],\quad n\geq 1,\nonumber
\end{align*}
which, together with \eqref{standing ass.}, Lemma \ref{lem3.1}, the monotone convergence theorem and the Fubini's theorem, leads to
\begin{align*}
\mathbb{E}\left[\sum_{i=1}^{\infty}e^{-\delta T_i}\left(\log\bigg(\frac{X_{T_i}}{X_{T_{i-1}}^{\gamma}}\bigg)\right)_{+}\right]
&\leq 
A^*+C^*\log x+ \mathbb{E}\left[\sum_{i=1}^{\infty}e^{-\delta T_i}\left(\log\bigg(\frac{X_{T_i}}{X_{T_{i-1}}^{\gamma}}\bigg)\right)_{-}\right]
\nonumber\\
&=A^*+C^*\log x+ \sum_{i=1}^{\infty}e^{-\delta T_i}\mathbb{E}\left[\left(\log\bigg(\frac{X_{T_i}}{X_{T_{i-1}}^{\gamma}}\bigg)\right)_{-}\right]
\nonumber\\
&<\infty,\quad x\in\mathbb{R}_+.\nonumber
\end{align*}
Therefore, the random variable $\sum_{i=1}^{\infty}e^{-\delta T_i}\left(\log\bigg(\frac{X_{T_i}}{X_{T_{i-1}}^{\gamma}}\bigg)\right)_{+}$ is integrable. Recalling that, by \eqref{standing ass.} and the Fubini's theorem, the random variable $\sum_{i=1}^{\infty}e^{-\delta T_i}\left(\log\bigg(\frac{X_{T_i}}{X_{T_{i-1}}^{\gamma}}\bigg)\right)_{-}$ is also integrable. In addition, one can verify that
\begin{align}\label{Control function}
    &\left|\sum_{i=1}^{n}e^{-\delta T_i}\left(\log\bigg(\frac{X_{T_i}}{X_{T_{i-1}}^{\gamma}}\bigg)\right)\right|\nonumber\\
\leq & \sum_{i=1}^{n}e^{-\delta T_i}\left(\log\bigg(\frac{X_{T_i}}{X_{T_{i-1}}^{\gamma}}\bigg)\right)_{+}
    +
    \sum_{i=1}^{n}e^{-\delta T_i}\left(\log\bigg(\frac{X_{T_i}}{X_{T_{i-1}}^{\gamma}}\bigg)\right)_{-}
    \nonumber\\
    \leq
&    \sum_{i=1}^{\infty}e^{-\delta T_i}\left(\log\bigg(\frac{X_{T_i}}{X_{T_{i-1}}^{\gamma}}\bigg)\right)_{+}
    +
    \sum_{i=1}^{\infty}e^{-\delta T_i}\left(\log\bigg(\frac{X_{T_i}}{X_{T_{i-1}}^{\gamma}}\bigg)\right)_{-},\quad n\geq 1,
\end{align}
where the random variable on the right hand side of \eqref{Control function} is integrable that
\begin{align}
\label{control func.integrable}
\mathbb{E}\left[\sum_{i=1}^{\infty}e^{-\delta T_i}\left(\log\bigg(\frac{X_{T_i}}{X_{T_{i-1}}^{\gamma}}\bigg)\right)_{+}
    +
    \sum_{i=1}^{\infty}e^{-\delta T_i}\left(\log\bigg(\frac{X_{T_i}}{X_{T_{i-1}}^{\gamma}}\bigg)\right)_{-}\right]<\infty.
\end{align}
By \eqref{control func.integrable}, we know that
\begin{align*}
    &\sum_{i=1}^{\infty}e^{-\delta T_i}\left|\log\bigg(\frac{X_{T_i}}{X_{T_{i-1}}^{\gamma}}\bigg)\right|
\nonumber\\
=&
     \sum_{i=1}^{\infty}e^{-\delta T_i}\left(\log\bigg(\frac{X_{T_i}}{X_{T_{i-1}}^{\gamma}}\bigg)\right)_{+}
    +
    \sum_{i=1}^{\infty}e^{-\delta T_i}\left(\log\bigg(\frac{X_{T_i}}{X_{T_{i-1}}^{\gamma}}\bigg)\right)_{-}<\infty\,\,  \text{a.s.}.\nonumber
\end{align*}
Because the absolute convergent series is conditionally convergent, one has
\begin{align}
\label{convergence.a.s.}
    \sum_{i=1}^{n}e^{-\delta T_i}\log\bigg(\frac{X_{T_i}}{X_{T_{i-1}}^{\gamma}}\bigg)\longrightarrow
    \sum_{i=1}^{\infty}e^{-\delta T_i}\log\bigg(\frac{X_{T_i}}{X_{T_{i-1}}^{\gamma}}\bigg) \text{ a.s. when}\ n\rightarrow\infty.
\end{align}
With the help of \eqref{Control function}, \eqref{control func.integrable} and \eqref{convergence.a.s.}, one can apply Lemma \ref{lem3.1} and the dominated convergence theorem to the inequality of \eqref{supmarleadsto} to conclude that
\begin{eqnarray}
\mathbb{E}\left[\sum_{i=1}^{\infty}e^{-\delta T_i}\log\bigg(\frac{X_{T_i}}{X_{T_{i-1}}^{\gamma}}\bigg)\right]\leq A^*+C^*\log x,\nonumber
\end{eqnarray}
and thus it follows that
\begin{eqnarray}
V(x)=\sup_{\pi\in\mathcal{U}_0(x)}\mathbb{E}\left[\sum_{i=1}^{\infty}e^{-\delta T_i}\log\bigg(\frac{X_{T_i}}{X_{T_{i-1}}^{\gamma}}\bigg)\right]\leq A^*+C^*\log x,\quad x\in\mathbb{R}_+.\nonumber
\end{eqnarray}
For the reverse inequality, it suffices to show the existence of some admissible process $X^*$ such that $$\mathbb{E}\left[\sum_{i=1}^{\infty}e^{-\delta T_i}\log\left(\frac{X^*_{T_i}}{(X^*_{T_{i-1}})^{\gamma}}\right)\right]=A^*+C^*\log x,$$
as the proof of Theorem \ref{thm4.1}.
To prove the second claim, we only need to verify that ${X}^*_t=x\frac{B_{t}}{Z^{\Tilde{\pi}^*}_{t}}$ satisfies 
$d{X}^*_t=[r{X}^*_t+(\pi^*_t)^{\mathsf{T}}(\mu-r\mathbf{1})]dt+(\pi^*_t)^{\mathsf{T}}\sigma dW_t$. To this end, we observe that
\begin{eqnarray}
d\left[\frac{B_{t}}{Z^{\Tilde{\pi}^*}_t}\right]=\frac{B_{t}}{Z^{\Tilde{\pi}^*}_t}\left[\left(r+\|\xi+\sigma^{\mathsf{-1}}{\Tilde{\pi}^*_t}\|^2\right)dt+(\xi+\sigma^{-1}{\Tilde{\pi}^*_t})^{\mathsf{T}}dW_t\right],\quad t\in[0,\tau].\nonumber
\end{eqnarray}
Therefore
\begin{eqnarray}\label{dX*}
d{X}^*_t=d\left[x\frac{B_t}{Z^{\Tilde{\pi}^*}_t}\right]
\hspace{-0.3cm}&=&\hspace{-0.3cm}
{X}^*_t\left[\left(r+\|\xi+\sigma^{-1}{\Tilde{\pi}^*_t}\|^2\right)dt+(\xi+\sigma^{-1}{\Tilde{\pi}^*_t})^{\mathsf{T}}dW_t\right]
\nonumber\\
\hspace{-0.3cm}&=&\hspace{-0.3cm}
r{X}^*_tdt+(\pi_t^*)^{\mathsf{T}}(\mu-r\mathbf{1})+(\pi^*_t)^{\mathsf{T}}{\Tilde{\pi}^*_t}dt+(\pi^*_t)^{\mathsf{T}}\sigma dW_t.
\end{eqnarray}
Furthermore, one can verify that for some $\vartheta\in(0,1)$, $\varrho\in(1,\infty)$ and any $x\in\mathbb{R}_+$, it holds that $\vartheta \log^{\prime}(x)\geq \log^{\prime}(\varrho x)$. Recall that $\Tilde{\pi}^*$ is the optimal solution to the dual problem \eqref{dual.pro.log}. Then, using Theorem 4.8 of \cite{XS92(a)},
one can verify that $(\pi^*_t)^{\mathsf{T}}\Tilde{\pi}^*_t=0$, which together with \eqref{dX*} verifies the second claim.
\end{proof}

\ \\
\begin{proof}[Proof of Proposition \ref{optimaltau-1}]
Note that
\begin{eqnarray}\label{3.50}
\hspace{-0.5cm}
V(x;\tau)
\hspace{-0.3cm}&=&\hspace{-0.3cm}
\frac{1-\gamma}{e^{\delta\tau}-1}\left[\left(\frac{e^{\delta\tau}-\gamma}{1-\gamma}\frac{\delta\tau}{e^{\delta\tau}-1}-1\right)\frac{r+\frac{1}{2}\|\Tilde{\xi}\|^2}{\delta}+\log x+\frac{r+\frac{1}{2}\|\Tilde{\xi}\|^2}{\delta}\right],\quad(x,\tau)\in\mathbb{R}_+^{2},
\end{eqnarray}
which, together with the fact that $\lim_{\tau\rightarrow0+}\left(\frac{e^{\delta\tau}-\gamma}{1-\gamma}\frac{\delta\tau}{e^{\delta\tau}-1}-1\right)=0$ and the assumption $\frac{r+\frac{1}{2}\|\Tilde{\xi}\|^2}{\delta}+\log x<0$, implies that 
\begin{eqnarray}
\lim_{\tau\rightarrow0+}V(x;\tau)\frac{e^{\delta\tau}-1}{1-\gamma}=\frac{r+\frac{1}{2}\|\Tilde{\xi}\|^2}{\delta}+\log x<0.\nonumber
\end{eqnarray}
The above inequality and the fact that $\lim\limits_{\tau\rightarrow0+}\frac{e^{\delta\tau}-1}{1-\gamma}=0$ give rise to
\begin{eqnarray}\label{V(tau).1}
\lim_{\tau\rightarrow0+}V(x;\tau)=-\infty.
\end{eqnarray}
In the sequel, we verify that there exists some $\tau_0\in\mathbb{R}_+$ such that $V(x;\tau_0)>0$. In fact, one knows that
$$\lim_{\tau\rightarrow\infty}\left(\frac{e^{\delta\tau}-\gamma}{1-\gamma}\frac{\delta\tau}{e^{\delta\tau}-1}-1\right)=\infty,$$
which together with $r+\frac{1}{2}\|\Tilde{\xi}\|^2>0$ and \eqref{3.50} yields the existence of some $\tau_0\in\mathbb{R}_+$ satisfying
$$V(x;\tau_0)\frac{e^{\delta\tau_0}-1}{1-\gamma}>0,$$
and hence $V(x;\tau_0)>0$.
This, together with \eqref{V(tau).1} and the fact that $\lim_{\tau\rightarrow\infty}V(x;\tau)=0$, implies the desired result. 
\end{proof}

\ \\
\begin{proof}[Proof of Proposition \ref{optimaltau-2}]
For any fixed $x\in\mathbb{R}_+$, let us consider
\begin{eqnarray}\label{f(tau).2}
f(\tau):=V(x;\tau)\tau=\frac{e^{\delta\tau}-\gamma}{(e^{\delta\tau}-1)^2}\left(r+\frac{1}{2}\|\Tilde{\xi}\|^2\right)\tau^2+\frac{(1-\gamma)\tau}{e^{\delta\tau}-1}\log x,\quad \tau\in\mathbb{R}_+.
\end{eqnarray}
For the first claim with $\gamma=1$, the above equation can be reduced to
\begin{eqnarray}\label{f(tau).1}
f(\tau)=\frac{\tau^2}{e^{\delta\tau}-1}\left(r+\frac{1}{2}\|\Tilde{\xi}\|^2\right),\quad \tau\in\mathbb{R}_+.
\end{eqnarray}
Differentiating the both sides of \eqref{f(tau).1} gives
\begin{eqnarray}
f^{\prime}(\tau)=\frac{(2e^{\delta\tau}-2-\tau\delta e^{\delta\tau})\tau}{(e^{\delta\tau}-1)^2}\left(r+\frac{1}{2}\|\Tilde{\xi}\|^2\right),\quad\tau\in\mathbb{R}_+.\nonumber
\end{eqnarray}
One can verify that the function $\mathbb{R}_{+}\ni\tau\mapsto 2e^{\delta\tau}-2-\tau\delta e^{\delta\tau}$ is increasing on $(0,1/\delta)$ and decreasing on $(1/\delta,\infty)$ with $\lim_{\tau\rightarrow0+}(2e^{\delta\tau}-2-\tau\delta e^{\delta\tau})=0$ and $\lim_{\tau\rightarrow\infty}(2e^{\delta\tau}-2-\tau\delta e^{\delta\tau})=-\infty$. Hence, it follows that there exists some $\tau^*\in(1/\delta,\infty)$ such that $f^{\prime}(\tau)$ is positive on $(0,\tau^*)$ and negative on $(\tau^*,\infty)$, implying that $f(\tau)$ is increasing on $(0,\tau^*)$ and decreasing on $(\tau^*,\infty)$ with $\lim_{\tau\rightarrow0+}f(\tau)=\lim_{\tau\rightarrow\infty}f(\tau)=0$. 
Therefore, the first claim holds true. 

For the second claim, we assume $\gamma\in(0,1)$ and $\left(r+\frac{1}{2}\|\Tilde{\xi}\|^2\right)\frac{\gamma}{\delta}-\frac{1-\gamma}{2}\log x>0$. Differentiating the both sides of \eqref{f(tau).2} gives
\begin{eqnarray}
f^{\prime}(\tau)
\hspace{-0.3cm}&=&\hspace{-0.3cm}
\frac{e^{2\delta\tau}(2-\delta\tau)\tau+e^{\delta\tau}[-\delta\tau+2\delta\gamma\tau-2(1+\gamma)]\tau+2\gamma\tau}{(e^{\delta\tau}-1)^3}\left(r+\frac{1}{2}\|\Tilde{\xi}\|^2\right)
\nonumber\\
\hspace{-0.3cm}&&\hspace{-0.3cm}
+\frac{e^{\delta\tau}-1-\tau\delta e^{\delta\tau}}{(e^{\delta\tau}-1)^2}(1-\gamma)\log x,\nonumber
\end{eqnarray}
which combined with the facts that $r+\frac{1}{2}\|\Tilde{\xi}\|^2>0$, $\delta>0$ and $\gamma\in(0,1)$ yields the existence of some $\tau_0\in(0,\infty)$ such that $f^{\prime}(\tau)\leq0$ for all $\tau\in[\tau_0,\infty)$. 
Therefore, there exists some $\tau^*\in[0,\tau_0]$ satisfying $f(\tau^*)\geq f(\tau)$ for $\tau\in[0,\infty)$. One can also verify that $\tau^*\in(0,\tau_0]$.
In fact, according to the assumption in the second claim, it can be seen that 
\begin{align}
\lim_{\tau\rightarrow0+}f^{\prime}(\tau)&=\frac{\gamma}{\delta}\left(r+\frac{1}{2}\|\Tilde{\xi}\|^2\right)-\frac{1-\gamma}{2}\log x>0,\nonumber
\end{align}
which implies that $f(\tau)$ attains its maximum at $\tau^*\in(0,\tau_0]$. The proof is complete.
\end{proof}

\ \\
\ \\
\textbf{Acknowledgements}:  W. Wang is supported by the National Natural Science Foundation of China under no. 12171405 and no. 11661074.
X. Yu is supported by the Hong Kong RGC General Research Fund (GRF) under grant no. 15304122.

\ \\

\end{document}